\documentclass[10pt,twocolumn,journal]{IEEEtran}
\usepackage{latexsym}
\usepackage{amsfonts}
\usepackage{amsbsy}
\usepackage{amsmath,amssymb}
\usepackage{times}
\usepackage{graphicx}
\usepackage{enumerate}
\usepackage[usenames]{color}
\usepackage[dvips]{pstcol}
\usepackage{epstopdf}

\usepackage{subeqnarray}
\usepackage{cases}
\usepackage{stfloats}
\usepackage{subeqnarray}
\usepackage{amsmath} 
\usepackage{multicol}
\usepackage{multirow}
\usepackage{bigstrut}
\usepackage{amsthm}
\usepackage{booktabs}
\usepackage{bm}
\usepackage{algorithm}
\usepackage{algorithmic}
\usepackage{bm}
\usepackage{amsthm}
\usepackage{cite}
\input epsf

\title{Exploiting Amplitude Control in Intelligent Reflecting Surface Aided Wireless Communication with Imperfect CSI}
\author{Ming-Min Zhao, \IEEEmembership{Member,~IEEE,} Qingqing Wu, \IEEEmembership{Member,~IEEE,} Min-Jian Zhao, \IEEEmembership{Member,~IEEE,} and Rui Zhang, \IEEEmembership{Fellow,~IEEE}
	\thanks{
		M. M. Zhao and M. J. Zhao are with the College of Information Science and Electronic Engineering, Zhejiang University (email: \{zmmblack, mjzhao\}@zju.edu.cn). Q. Wu is with the State Key Laboratory of Internet of Things for Smart City, University of Macau, Macau, 999078, and also with the National Mobile Communications Research Laboratory, Southeast University, Nanjing 210096, China (email: qingqingwu@um.edu.mo). 
		  R. Zhang is with the Department of Electrical and Computer Engineering, National University of Singapore (email: elezhangg@nus.edu.sg). This article was presented in part at the IEEE Global Communications Conference 2020 \cite{Zhaoglobecom2020}.
	}
}

\begin{document}
		\maketitle
	\begin{abstract} 
Intelligent reflecting surface (IRS) is a promising new paradigm to achieve high spectral and energy efficiency for future wireless networks by reconfiguring the wireless signal  propagation via passive reflection. To reap the promising gains of IRS, channel state information (CSI) is essential, whereas channel estimation errors are inevitable in practice due to limited channel training resources. In this paper, in order to optimize the performance of IRS-aided multiuser communications with imperfect CSI, we propose to jointly design the active transmit precoding at the access point (AP) and passive reflection coefficients of the IRS, each consisting of not only the conventional phase shift and also the newly exploited amplitude variation. First, the achievable rate of each user is derived assuming a practical IRS channel estimation method, which shows that the interference due to CSI errors is intricately related to the AP transmit precoders, the channel training power and the IRS reflection coefficients during both channel training and data transmission. Next, for the single-user case, by combining the benefits of the penalty method, Dinkelbach method and block successive upper-bound minimization (BSUM) method, a new penalized Dinkelbach-BSUM algorithm is proposed to optimize the IRS reflection coefficients for maximizing the achievable data transmission rate subjected to CSI errors; while for the multiuser case, a new penalty dual decomposition (PDD)-based algorithm is proposed to maximize the users' weighted sum-rate. Finally, simulation results are presented to validate the effectiveness of our proposed algorithms as compared to benchmark schemes. In particular, useful insights are drawn to characterize the effect of IRS reflection amplitude control (with/without the conventional phase-shift control) on the system performance under imperfect CSI.

	\end{abstract} 

\section{Introduction}
Due to the proliferation of mobile devices and increasing demand for high-speed data applications, various advanced wireless technologies such as massive multiple-input multiple-output (MIMO), ultra-dense network (UDN) and millimeter wave (mmWave) communications, have been proposed and thoroughly investigated to improve the wireless communication network spectral efficiency \cite{Boccardi2014}. However, these technologies generally incur higher energy consumption and hardware cost, due to the ever-increasing number of active nodes/antennas/radio-frequency (RF) chains employed in the network. To alleviate this issue, intelligent reflecting surface (IRS) has been proposed recently as a promising new paradigm to achieve highly spectral-efficient, yet low-cost and low-energy wireless systems in the future \cite{Wu2019Magazine, Wu2018_journal, Basar2019, Huang2019}. Specifically, IRS is a man-made planar metasurface composed of a large number of passive reflecting elements, each of which is able to induce certain amplitude and/or phase changes in its reflected signal, thus collaboratively altering the signal propagation from the transmitter to receiver(s) to achieve various objectives, such as signal enhancement and interference suppression \cite{cui2014coding}. Due to its low cost, IRS can be densely deployed in wireless networks. In addition, different from conventional active relays, IRS can achieve full-duplex signal reflection without self-interference and processing noise.

By properly designing the IRS reflection coefficients, it has been shown that IRS can significantly enhance the performance of various wireless systems (see, e.g., \cite{Yang2019, Cui2019, Jiang2019, zuo2020resource, ZhangMIMO}). However, in order to achieve the performance gains offered by IRS, the acquisition of accurate channel
state information (CSI) at the IRS, for the links with both its associated access point (AP) and users, is crucial, which however is practically difficult due to the passive nature of IRS and its large number of reflecting elements. In the literature, various methods have been proposed to efficiently estimate the IRS channels \cite{Mishra2019ICASSP, zheng2019intelligent, you2019progressive, wang2019channel, zheng2020intelligent, He2019_CE, chen2019channel}. Specifically, in \cite{Yang2019} and \cite{Mishra2019ICASSP}, an on/off reflection control based least-square (LS) channel estimation method was proposed, where only one IRS element is switched on to estimate the corresponding reflected channel  at each time. To exploit the IRS's large aperture in channel estimation, a discrete Fourier transform (DFT) reflection pattern based channel estimation method was proposed in \cite{zheng2019intelligent}, where the reflection amplitudes of all IRS elements are set to the maximum value of unity. In \cite{you2019progressive}, discrete phase shifts at the IRS were considered and the reflection patterns for channel estimation were designed under this practical constraint. In \cite{wang2019channel} and \cite{zheng2020intelligent}, IRS-aided multiuser system was considered and it was shown that the IRS channel training overhead can be effectively reduced by exploiting the fact that each IRS element reflects the signals from different users to the AP via the same IRS-AP channel. Besides, channel properties such as low-rank and sparsity were exploited in \cite{He2019_CE, chen2019channel} for IRS channel estimation.

Despite the progress in channel estimation for IRS-aided systems, channel estimation errors are inevitable in practice due to the limited channel training resources (such as power and time), which result in performance degradation. Therefore, it is crucial to take them into account when designing IRS reflections for data transmission. However, there are only few works that have studied robust IRS designs under imperfect CSI \cite{yu2019robust, zhou2019robust, zhou2020framework}. The robust designs in these works are based on certain canonical CSI error models, e.g., bounded error and statistical error models, while in practice the distribution of CSI errors depends on the specific channel estimation method adopted. Besides, existing works mainly exploit the IRS phase-shift control, while assuming full amplitude reflection of its elements. Thus, the effect of IRS amplitude control on its performance, especially under CSI errors, is unexploited yet to the authors' best knowledge. It is worth noting that with perfect CSI, IRS full reflection has been largely assumed in the existing literature to maximize the reflected signal power by IRS \cite{ Wu2018_journal, Wu2019Discrete, you2019progressive}. This is reasonable for the case with one single-antenna user, since the phase shift of each reflecting element can be adjusted such that the reflected signals from all reflecting elements are added constructively at the user receiver, even under the practical constraint with discrete phase shifts \cite{Wu2019Discrete}. Moreover, for the IRS-aided multiuser system with co-channel interference, it was shown in \cite{guo2019weighted} that the performance gain offered by amplitude control is almost negligible under perfect CSI. However, under imperfect CSI, it remains unknown whether exploiting the IRS reflection amplitude control is beneficial or not, which motivates this work.

In this paper, we consider an IRS-aided multiuser multiple-input single-output (MISO) system, where the reflection coefficients at the IRS (including both reflection amplitudes and phase shifts) and the active transmit precoders at the multi-antenna AP are jointly optimized to maximize the achievable rates of a set of single-antenna users under imperfect CSI. To characterize the distribution of CSI errors, the IRS channel estimation methods in \cite{zheng2019intelligent, you2019progressive} are considered, for which the statistics of CSI errors are obtained  accordingly. Then, tractable lower bounds of the mutual information between the transmit symbols at the AP and the received signals at the users, i.e., their achievable rates, are derived. It is shown that CSI errors cause additional interference that is intricately related to the AP transmit precoders, the uplink training power, and the IRS reflection designs in both channel training and data transmission. For maximizing the achievable rates, two new algorithms are proposed for the single-user and multiuser cases, respectively. In particular, we first consider the single-user case and show why full IRS reflection is generally undesired under imperfect CSI. Then, by leveraging the penalty method \cite{Bertsekas1999}, Dinkelbach method \cite{dinkelbach1967nonlinear} and block successive upper-bound minimization (BSUM) method \cite{Hong2016}, we propose a new penalized Dinkelbach-BSUM algorithm for solving the optimization problem efficiently, and prove its convergence. Next, for the general multiuser case, we propose a penalty dual decomposition (PDD)-based algorithm (similar to that in \cite{shi2017penalty, zhao2019intelligent}) to maximize the users' weighted sum-rate. Both algorithms can be easily modified to handle discrete/continuous IRS amplitude and/or phase-shift cases in practice. Numerical results validate the effectiveness of the proposed algorithms and show that by controlling IRS reflection amplitude under imperfect CSI, additional performance gains can be achieved over the conventional schemes with full reflection. Moreover, we show that under certain practical setups, controlling IRS reflection amplitude is more cost-effective than phase shift.

To the best of our knowledge, this is the first work on exploiting the IRS amplitude control for IRS-aided communication systems under {\it imperfect CSI} and the new contributions of this paper over the existing literature are summarized as follows:

1) We establish a general optimization framework for joint active and passive beamforming design in an IRS-aided communication system based on a practical CSI error model, where the IRS reflection amplitude is exploited for performance enhancement in addition to the conventional phase shift.

2) Two new algorithms are proposed for the single-user and multiuser cases, respectively, and both algorithms can handle continuous and discrete reflection amplitude/phase-shift cases and constitute efficient variable updating steps, which either admit closed-form solutions or can be carried out via simple iterative procedures.

3) Extensive numerical results are presented to validate the effectiveness of the proposed scheme with amplitude control and useful insights are drawn.

The rest of the paper is organized as follows. In Section \ref{section1}, we present the system model, channel estimation method and problem formulation. In Sections
\ref{Section_SU} and \ref{Section_MU}, we propose efficient algorithms to solve the formulated problems in the single-user and multiuser cases, respectively. In Section \ref{Section_Simulation}, numerical results are provided to evaluate the performance of the proposed algorithms. Finally, we conclude the paper in Section \ref{Section_conclusion}.

\emph{Notations}: Scalars, vectors and matrices are respectively denoted by lower/upper case, boldface lower case and boldface upper case letters. For an arbitrary matrix $\mathbf{A}$, $\mathbf{A}^T$, $\mathbf{A}^*$, $\mathbf{A}^{H}$ and $\mathbf{A}^{\dagger} $ denote its transpose, conjugate, conjugate transpose and pseudo-inverse, respectively. $\mathbf{A}^{-1}$ denotes the inverse of a square matrix $\mathbf{A}$ if it is invertible. $\mathbb{C}^{n\times m}$ denotes the space of $n\times m$ complex matrices. For matrices $\mathbf{A} \in \mathbb{C}^{N_1 \times M}$ and $\mathbf{B} \in \mathbb{C}^{N_2 \times M}$, $[\mathbf{A};\mathbf{B}] \in \mathbb{C}^{(N_1+N_2)\times M}$ denotes row-wise concatenation of $\mathbf{A}$ and $\mathbf{B}$. $\mathbf{A}_{mn,k}$ ($\mathbf{A}_{mn}$) denotes the element on the $m$-th row and $n$-column of matrix $\mathbf{A}_k$ ($\mathbf{A}$). $\|\cdot\|$ and $\|\cdot\|_{\infty}$ denote the Euclidean norm and infinity norm of a complex vector, respectively, and $|\cdot|$ denotes the absolute value of a complex scalar or the cardinality of a finite set. $\mathbf{a} \cdot \mathbf{b}$ denotes the dot product of two vectors. $\mathcal{CN}(\mathbf{x},\bm{\Sigma})$ denotes the distribution of a circularly symmetric complex Gaussian (CSCG) random vector with mean vector $\mathbf{x}$ and covariance matrix $\bm{\Sigma}$; and $\sim$ stands for ``distributed as''. For given numbers $x_1,\cdots,x_N$, $\textrm{diag}(x_1,\cdots,x_N)$ denotes a diagonal matrix with $\{x_1,\cdots,x_N\}$ being its diagonal elements and $\textrm{diag}(\mathbf{A})$ denotes a vector which contains the diagonal elements of matrix $\mathbf{A}$. The symbol $\jmath$ is used to represent $\sqrt{-1}$. For a complex number $x$, $\Re \{x\}$ denotes its real part and $\angle x$ denotes its angle. $\mathbf{I}$ and $\mathbf{0}$ denote an identity matrix and an all-zero vector with appropriate dimensions, respectively. $\mathbb{E}\{\cdot\}$ denotes the statistical expectation. The set difference is defined as $\mathcal{A}\backslash \mathcal{B} \triangleq \{x| x\in\mathcal{A},x\notin \mathcal{B}\}$. 

\section{System Model and Problem Formulation} \label{section1}

\subsection{System Model}
As shown in Fig. \ref{fig:figure1}, we consider an IRS-aided multiuser MISO downlink communication system, where an IRS composed of $N$ passive reflecting elements is deployed to assist in the communication from the AP to a set of $K$ users denoted by $\mathcal{K}\triangleq\{1,\cdots, K\}$. We assume that the AP is equipped with $M$ transmit antennas, and each user is equipped with a single antenna. The IRS is attached with a smart controller, which is  connected with the AP via a separate reliable wireless link and responsible for coordinating their operation as well as exchanging information such as reflection coefficients and CSI. The signals reflected by IRS two or more times are ignored due to the severe ``distance-product'' power loss over multiple reflections \cite{Wu2018_journal}.

\begin{figure}[!hhh] 
	\centering
	\scalebox{0.4}{\includegraphics{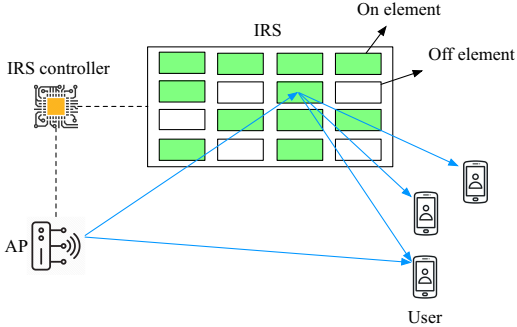}} 
	\caption{An IRS-aided multiuser MISO downlink system with on/off reflection.} 
	\label{fig:figure1}
\end{figure}

We consider quasi-static block-fading channels and all channels are assumed to remain approximately constant in each fading block. Let $\mathbf{h}_{d,k} \in \mathbb{C}^{M\times 1}$ with $k\in\mathcal{K}$ denote the direct (conjugate) channel vector of the AP-user $k$ link, $\mathbf{G} \in \mathbb{C}^{N\times M}$ denote the channel matrix of the AP-IRS link and $\mathbf{h}_{r,k} \in \mathbb{C}^{N\times 1}$ denote the (conjugate) channel vector of the IRS-user $k$ link. Thus, the received signal of user $k$ is expressed as
\begin{equation} \label{received_signal}
\begin{array}{l}
y_k =   (\mathbf{h}_{r,k}^H \mathbf{\Theta} \mathbf{G} + \mathbf{h}_{d,k}^H)\sum\limits_{j\in\mathcal{K}}\mathbf{w}_js_j + n_k,
\end{array}
\end{equation}
where $s_k\sim \mathcal{CN}(0, 1)$ denotes the transmit symbol for user $k$, with $s_k$'s assumed to be independent and identically distributed (i.i.d.); $\mathbf{w}_k \in \mathbb{C}^{M\times 1}$ represents the transmit precoder for user $k$; $n_k$ is the additive white Gaussian noise (AWGN) at user $k$ with zero-mean and variance $\sigma_k^2$; $\bm{\Theta}=\textrm{diag}(\phi_1,\cdots, \phi_n,\cdots,\phi_N)$ ($n \in \mathcal{N} \triangleq \{1,\cdots, N\}$) denotes the reflection-coefficient matrix at the IRS with $\phi_n = a_n e^{\jmath \theta_n}$, where $a_n \in [0, 1]$ and $\theta_n \in [0,2\pi)$ represent the reflection amplitude and phase shift of the $n$-th element,
respectively. In practice, due to hardware limitations, the amplitude and phase shift of each reflecting element can only take a finite number of discrete values \cite{Wu2019Discrete, Wu2019Magazine}. Besides, different from most existing works that assume $a_n=1,\; \forall n \in \mathcal{N}$, i.e., each element is designed to maximize the signal reflection \cite{Wu2018_journal, Wu2019Discrete}, we consider joint amplitude and phase-shift control in this paper to unveil the full benefit offered by IRS. Note that independent control over the reflection amplitude and phase shift of each IRS element can be achieved in practice, e.g., by equipping the reflecting elements with locally tunable integrated circuits (ICs) that provide a continuously tunable complex impedance \cite{PhysRevApplied}, or through multilayer surface design \cite{Nayeri2018book}. Let $Q_a$ and $Q_\theta$ denote the number of bits for reflection amplitude and phase-shift control per IRS element, respectively. We thus have 
\begin{equation} \label{A_Ps_constraint}
\phi_n \in \mathcal{F}_d \triangleq \{\phi_n| \phi_n = a_n e^{\jmath \theta_n}, \theta_n \in \mathcal{S}, a_n \in \mathcal{A} \}, 
\end{equation}
where $\mathcal{S} \triangleq \{0,\frac{2\pi}{L},\cdots, \frac{2\pi(L-1)}{L}\}$ with $L=2^{Q_\theta}$, i.e., the discrete phase-shift values are assumed to be equally spaced in the interval $[0,2\pi)$, and $\mathcal{A} \triangleq \{\bar{a}_1,\cdots, \bar{a}_{2^{Q_a}} \}$ denotes the controllable amplitude set which satisfies $|\mathcal{A}| = 2^{Q_a}$. Note that when $Q_a =0$, $\mathcal{A}$ reduces to the case of full reflection, i.e., $\mathcal{A} = \{1\}$, while when $Q_a =1$, $\mathcal{A}$ represents on/off reflection, i.e., $\mathcal{A} = \{0,1\}$, as shown in Fig. \ref{fig:figure1}. Furthermore, by letting $Q_\theta \rightarrow \infty $ and $Q_a  \rightarrow \infty$, the model in \eqref{A_Ps_constraint} becomes the case with continuous amplitude/phase shift, i.e., $\phi_n \in \mathcal{F}_c \triangleq \{\phi_n| \phi_n = a_n e^{\jmath \theta_n}, \theta_n \in [0,2\pi ), a_n \in [0, 1] \}$, which leads to the performance upper bound for discrete amplitude/phase-shift cases in practical systems.

Note that the reflection model of IRS depends on the specific design of the reflecting elements that interact with the incoming electromagnetic waves. In \cite{Abeywickrama2020, jung2019optimality}, a coupled reflection amplitude and phase shift model was proposed by modeling each reflecting element as a resonant circuit with certain inductance, capacitance, and resistance. Compared to the model in \cite{Abeywickrama2020, jung2019optimality}, the reflection model considered in this paper with independent control of the reflection amplitude and phase shift simultaneously is more flexible, but it also requires more sophisticated element design \cite{Wu2020tutorial}.

\subsection{Channel Estimation and Achievable Rate} \label{CE_model}
To achieve the joint active and passive beamforming gain offered by the IRS-aided system, accurate CSI at the AP/IRS is crucial, which however is practically difficult to obtain. In this paper, we assume that the downlink-uplink channel reciprocity holds, thus the downlink channel can be learned by estimating its counterpart in the uplink by varying the IRS reflection patterns \cite{zheng2019intelligent, you2019progressive}. Note that for multiuser systems, the uplink channel estimation overhead can be reduced by exploiting the fact that the IRS reflects the simultaneously transmitted pilot signals from all users to the AP via the same IRS-AP channel \cite{wang2019channel, zheng2020intelligent}. However, for ease of exposition and simplicity, we assume in this paper that the user channels are estimated consecutively in the uplink. The proposed beamforming designs in this paper are applicable to other channel estimation methods provided that the first and second order statistics of the channel estimation errors are available.

For completeness, we extend the time-varying reflection pattern based channel estimation method in \cite{zheng2019intelligent, you2019progressive} to our considered system as follows, where the reflection amplitudes are set to $a_n=1,\;\forall n$  to maximize the signal reflection for channel estimation. Specifically, during the channel training phase, user $k$  sends $N_r$ $(N_r \geq N+1)$ pilot symbols to the AP\footnote{If the IRS elements grouping scheme in \cite{Yang2019} is used, $N_r$ should be no less than $G+1$, where $G$ denotes the number of groups. In this paper, we do not consider IRS elements grouping for simplicity while the proposed algorithms can be applied with elements grouping as well. }, while the IRS phase shifts are varied over pilot symbols according to some pre-designed pattern. Let $\bar{\mathcal{N}} \triangleq \mathcal{N} \cup\{N+1,\cdots,N_r\}$. By stacking the received uplink signals at the AP $\{ \mathbf{y}_{u,k}^n\}_{n \in \bar{\mathcal{N}}}$ into $\mathbf{Y}_{k}  = [\mathbf{y}_{u,k}^1,\cdots,\mathbf{y}_{u,k}^{N_r} ]$ and applying the LS estimation, the estimated channel can be expressed as\footnote{When $M=1$ and $N_r = N+1$, the channel estimation becomes identical to that proposed in \cite{you2019progressive}.}
\begin{equation}
\bar{\mathbf{H}}_k = \left(\frac{1}{\sqrt{p_{u,k}} s_{u,k}}\mathbf{Y}_{k} \mathbf{V}^{\dagger}\right)^H =\tilde{\mathbf{H}}_k +   \frac{1}{\sqrt{p_{u,k}} s_{u,k}}  (\mathbf{V}^{\dagger})^H \mathbf{N}_{u,k}^H,
\end{equation}
where $s_{u,k}$ denotes the uplink training symbol which is assumed to be $1$ without loss of optimality; $p_{u,k}$ is the uplink training signal power of user $k$; 
$\mathbf{N}_{u,k} = [\mathbf{n}^1_{u,k},\cdots,\mathbf{n}^{N_r}_{u,k}]$ with $\mathbf{n}^n_{u,k} \sim \mathcal{CN}(0,\varepsilon_{k}^2 \mathbf{I})$ denoting the uplink AWGN vector; $\tilde{{\mathbf{H}}}_k \triangleq \left[\mathbf{h}_{d,k}^H; \mathbf{H}_k\right] \in \mathbb{C}^{(N+1)\times M}$, where $\mathbf{H}_{k} \triangleq \textrm{diag}(\mathbf{h}_{r,k}^H) \mathbf{G}$ denotes the cascaded AP-IRS-user $k$ channel without IRS reflection; $\mathbf{V} = \left[\tilde{\mathbf{v}}_1,\cdots,\tilde{\mathbf{v}}_{N_r}\right]$ with $\tilde{\mathbf{v}}_n \triangleq \left[ 1, \mathbf{v}_n^T\right]^T$ denotes the set of extended reflection vectors and $\mathbf{v}_n= \textrm{diag}\{ \bm{\Theta}_n^*\}$ represents the reflection pattern employed at the $n$-th training symbol duration; $\bar{\mathbf{H}}_k = [ \hat{\mathbf{h}}_{d,k}^H; \hat{\mathbf{H}}_k] $, $ \hat{\mathbf{h}}_{d,k}$ and $\hat{\mathbf{H}}_k$ denote the estimates of the direct channel and the cascaded AP-IRS-user $k$ channel, respectively. Note that for the continuous phase-shift case, $\{ \tilde{\mathbf{v}}_n\}$ can be chosen from the columns of the $(N+1)$ DFT matrix \cite{zheng2019intelligent}, while for the discrete phase-shift case, $\{ \tilde{\mathbf{v}}_n\}$ can be chosen to be the columns of the quantized DFT matrix or truncated Hadamard matrix (THM) according to the value of $Q_{\theta}$ \cite{you2019progressive}. Therefore, for the discrete phase-shift case, the CSI error matrix $\Delta \tilde {\mathbf{H}}_k = \bar{\mathbf{H}}_k - \tilde{\mathbf{H}}_k$ satisfies $\mathbb{E} \{ \Delta \tilde {\mathbf{H}}_k\} = \mathbf{0}$ and $\mathbb{E} \{ \Delta \tilde {\mathbf{H}}_k \Delta \tilde {\mathbf{H}}_k^H  \}  = \frac{M\varepsilon_{k}^2}{p_{u,k}} (\mathbf{V} \mathbf{V}^H)^{\dagger}$. In other words, if the extended reflection vectors are non-orthogonal due to discrete phase shifts and/or $N_r>N+1$, then the CSI errors of different channel coefficients in the direct channel and the cascaded channel are correlated. In the special case of continuous phase shifts with $N_r=N+1$, we have $\Delta \mathbf{h}_{d,k} = \hat{\mathbf{h}}_{d,k} - \mathbf{h}_{d,k} \sim \mathcal{CN}(\mathbf{0}, \delta_{d,k}^2\mathbf{I})$ and $\Delta \mathbf{H}_{k} = \hat{\mathbf{H}}_k - \mathbf{H}_k \sim \mathcal{CN} (\mathbf{0}, \delta_{h,k}^2\mathbf{I}) $, where $\delta_{d,k}^2 =\delta_{h,k}^2= \frac{\varepsilon_{k}^2}{(N+1) p_{u,k}}$, i.e., the elements of the CSI errors are i.i.d. CSCG distributed.

Due to CSI errors, there is information loss as compared to the perfect CSI case and we need to find the user achievable rate that takes the CSI errors into consideration. Unfortunately, it is difficult to derive the mutual information $I(s_k; y_k|\hat{\mathbf{H}}, \hat{\mathbf{h}}_d )$ in closed-form, where $\hat{\mathbf{H}}= \{\hat{\mathbf{H}}_{k}\}_{k\in\mathcal{K}}$ and $\hat{\mathbf{h}}_d = \{\hat{\mathbf{h}}_{d,k}\}_{k\in\mathcal{K}} $. We therefore turn to finding a tractable lower bound on the mutual information $I(s_k; y_k|\hat{\mathbf{H}}, \hat{\mathbf{h}}_d )$, i.e., the achievable rate, which is given in the following proposition.
\newtheorem{prop}{Proposition}
\begin{prop} \label{achievable_rate_prop}
	\emph{
	$I(s_k; y_k|\hat{\mathbf{H}}, \hat{\mathbf{h}}_d ) $ is lower-bounded by 
	\begin{equation} \label{lower_bound}
	I(s_k; y_k|\hat{\mathbf{H}}, \hat{\mathbf{h}}_d )  \geq \log \left( 1 + {|(\mathbf{v}^H \hat{\mathbf{H}}_{k}+\hat{\mathbf{h}}_{d,k}^H)\mathbf{w}_k|^2}/{\Psi_k^d }\right),
	\end{equation}
	 where $\mathbf{v} = \textrm{diag}(\bm{\Theta}^*)$,
	\begin{equation} \label{inr_robust}
	\begin{aligned}
	\Psi_{k}^d 
	 \triangleq &  \sum\limits_{j \in \mathcal{K}\backslash k}  |(\mathbf{v}^H \hat{\mathbf{H}}_{k}  + \hat{\mathbf{h}}_{d,k}^H) \mathbf{w}_j|^2\\
	& +  \underbrace{(\bar{v}_{11,k}+\mathbf{v}^H\mathbf{r}_k+ \mathbf{r}_k^H\mathbf{v} + \mathbf{v}^H \mathbf{R}_k \mathbf{v}) \sum\limits_{j\in\mathcal{K}} \|\mathbf{w}_j\|^2}_{\textrm{interference due to imperfect CSI}}+  \sigma_k^2,
	\end{aligned}
	\end{equation}
	$\bar{\mathbf{V}}_k \triangleq  \left[ \begin{matrix}\bar{v}_{11,k} & \mathbf{r}_k^H\\ \mathbf{r}_k & \mathbf{R}_k
	\end{matrix} \right]$, $\bar{\mathbf{V}}_{ij,k}=\frac{\varepsilon_{k}^2}{p_{u,k}} (\check{\mathbf{v}}_{:,i}^* \cdot \check{\mathbf{v}}_{:,j})$ and $\check{\mathbf{v}}_{:,i}$ denotes the $i$-th column of $\mathbf{V}^{\dagger}$.
}
\end{prop}
\begin{proof} \vspace{-0.5em}
	Please refer to Appendix \ref{appendix_achievable_rate}.
\end{proof}
From \eqref{lower_bound} and \eqref{inr_robust}, we observe that for user $k$, the achievable rate is obtained by treating $\tilde{\mathbf{v}}\bar{\mathbf{H}} \mathbf{w}_ks_k$ as the desired signal, while the signals from the other users $\tilde{\mathbf{v}}\tilde{\mathbf{H}} \sum\nolimits_{j\in\mathcal{K}\backslash k}\mathbf{w}_js_j$ and the CSI-error-induced term $\tilde{\mathbf{v}} \Delta \tilde{\mathbf{H}}_k\mathbf{w}_ks_k$ are regarded as the interference. Besides, the interference due to imperfect CSI is related to the following four main factors: 1) the reflection pattern used for channel training, i.e., $\mathbf{V}$; 2) the uplink training power $\{p_{u,k}\}$; 3) the reflection coefficient vector $\mathbf{v}$ for data transmission; and 4) the AP transmit precoders $\{\mathbf{w}_k\} $.

\subsection{Problem Formulation}
In this paper, we aim to maximize the weighted sum of achievable rates of all users, by jointly optimizing the transmit precoders $\{\mathbf{w}_k\}$ at the AP and the reflection coefficient vector $\mathbf{v}$ at the IRS with imperfect CSI, subject to the total transmit power constraint at the AP as well as the IRS reflection amplitude/phase-shift constraints. The considered optimization problem can be formulated as
\begin{equation} \label{robust_problem_d}
	\begin{aligned}
	\max \limits_{\{\mathbf{w}_k \}, \mathbf{v}}\;  &\sum\limits_{k \in \mathcal{K}} \alpha_k \log \left( 1 + {|(\mathbf{v}^H \hat{\mathbf{H}}_{k}+\hat{\mathbf{h}}_{d,k}^H)\mathbf{w}_k|^2}/{\Psi_k^d }\right) \\
	\textrm{s.t.} \; & \sum\limits_{k \in \mathcal{K}} \|\mathbf{w}_k\|^2 \leq P, \\
	&  v_n \in \mathcal{F}_d ,\;\forall n \in \mathcal{N},
	\end{aligned}
\end{equation}
where $\alpha_k$ represents the weight of user $k$ and $P$ denotes the maximum transmit power at the AP.

In the special case with continuous amplitude/phase shift and $N_r=N+1$, we have $\bar{v}_{11,k} = \delta_{d,k}^2$, $\mathbf{r}_k = \mathbf{0}$ and $\mathbf{R}_k = \delta_{h,k}^2 \mathbf{I}$. Accordingly, problem \eqref{robust_problem_d} reduces to
	\begin{equation} \label{robust_problem_c}
	\begin{aligned}
	\max \limits_{\{\mathbf{w}_k \}, \mathbf{v}} \; & \sum\limits_{k \in \mathcal{K}} \alpha_k \log \left( 1 + {|(\mathbf{v}^H \hat{\mathbf{H}}_{k}+\hat{\mathbf{h}}_{d,k}^H)\mathbf{w}_k|^2}/{\Psi_k^c }\right) \\
	 \textrm{s.t.} \;& \sum\limits_{k \in \mathcal{K}} \|\mathbf{w}_k\|^2 \leq P,\\
	& v_n \in \mathcal{F}_c,\;\forall n \in \mathcal{N},
	\end{aligned}
\end{equation}
where
\begin{equation} \label{denominator}
\begin{aligned}
\Psi_{k}^c  \triangleq & \sum\limits_{j \in \mathcal{K}\backslash k}   |(\mathbf{v}^H \hat{\mathbf{H}}_{k} + \hat{\mathbf{h}}_{d,k}^H) \mathbf{w}_j|^2 \\
& +  \underbrace{ \delta_{d,k}^2\sum\limits_{j\in\mathcal{K}} \|\mathbf{w}_j\|^2 + \delta_{h,k}^2 \sum\limits_{j \in \mathcal{K}}   \|\mathbf{w}_j\|^2 \mathbf{v}^H \mathbf{v}}_{\textrm{interference due to imperfect CSI}} +  \sigma_k^2.\\
\end{aligned}
\end{equation}
From \eqref{denominator}, it is observed that in this case, the main factors that affect the interference due to imperfect CSI become: 1) the variance of CSI errors $\delta_{d,k}^2/\delta_{h,k}^2$ (depending on the uplink training power $\{p_{u,k} \}$ and the number of reflecting elements $N$); 2) the sum of squared reflection amplitudes $\mathbf{v}^H\mathbf{v}$, i.e., $\sum_{n \in \mathcal{N}} a_n^2$; and 3) the AP transmit precoders $\{\mathbf{w}_k\} $.

Both problems \eqref{robust_problem_d} and \eqref{robust_problem_c} are challenging to solve because their objective functions are non-concave as well as that the transmit precoders and the IRS reflection coefficients are non-linearly coupled. Problem \eqref{robust_problem_d} is more complex than \eqref{robust_problem_c} since the constraints $v_n \in \mathcal{F}_d,\;\forall n \in \mathcal{N} $ render it a mixed-integer nonlinear program (MINLP). Besides, compared to the conventional case with phase-shift control only, more optimization variables are involved in problems \eqref{robust_problem_d} and \eqref{robust_problem_c} as the IRS reflection amplitudes can also be optimized and the reflection amplitudes/phase shifts will have a joint effect on the system performance. In general, there are no efficient methods for solving the non-convex problems \eqref{robust_problem_d} and \eqref{robust_problem_c} optimally. In the next two sections, we propose  efficient algorithms to solve problem \eqref{robust_problem_d} sub-optimally in the single-user and multiuser cases, respectively, which can be applied to solve problem \eqref{robust_problem_c} as well.

\newtheorem{remark}{Remark}

\section{Single-User System} \label{Section_SU}
In this section, we consider the single-user case, i.e., $K = 1$, to draw useful insights into the gain of reflection amplitude control in addition to that of phase shift by IRS. 
In this case, multiuser interference does not exist, therefore we can simply drop the subscript $k$ and ignore the multiuser interference terms in problem \eqref{robust_problem_d}, which leads to the following optimization problem:
\begin{equation} \label{single_user_problem_d}
\begin{aligned}
\max \limits_{\mathbf{w},\;\mathbf{v}}\;& \log \left(1+ \frac{|(\mathbf{v}^H \hat{\mathbf{H}} + \hat{\mathbf{h}}_d^H)\mathbf{w}|^2}{ \|\mathbf{w}\|^2 (\mathbf{v}^H \mathbf{R} \mathbf{v}+\mathbf{v}^H\mathbf{r}+ \mathbf{r}^H\mathbf{v} + \bar{v}_{11}) + \sigma^2}\right) \\
\textrm{s.t.}\; & \|\mathbf{w}\|^2 \leq P,\\
& v_n \in \mathcal{F}_d,\;\forall n \in \mathcal{N}.
\end{aligned}
\end{equation}

To illustrate that amplitude control is helpful under imperfect CSI, we focus on the $n$-th reflection coefficient $v_n = a_n e^{-\jmath \theta_n}$ and assume that the optimal values of the other $\{v_j^{\textrm{opt}} \}_{j \in \mathcal{N}\backslash n}$ are given and fixed. Then, the signal power term in the objective function of problem \eqref{single_user_problem_d} can be equivalently rewritten as
\begin{equation}
\mathbf{v}^H \mathbf{A} \mathbf{v} + \mathbf{v}^H \mathbf{b} + \mathbf{b}^H \mathbf{v} +\hat{\mathbf{h}}_d^H \mathbf{w}\mathbf{w}^H \hat{\mathbf{h}}_d =
a_n^2 \mathbf{A}_{nn} +  a_n c_n + d_n,
\end{equation}
where $\mathbf{A} = \hat{\mathbf{H}} \mathbf{w} \mathbf{w}^H \hat{\mathbf{H}}^H$, $\mathbf{b} = \hat{\mathbf{H}} \mathbf{w} \mathbf{w}^H  \hat{\mathbf{h}}_d$, $c_n = \sum_{j\in\tilde{\mathcal{N}} \backslash n}  2 \Re\{e^{\jmath \theta_n} \mathbf{A}_{nj}  v_j^{\textrm{opt}} \} + 2  \Re\{  e^{\jmath \theta_n} b_n\}$, $\tilde{\mathcal{N}} \triangleq \mathcal{N} \cup\{N+1\}$ and $d_n = \sum_{i\in\tilde{\mathcal{N}}\backslash n} \sum_{j\in\tilde{\mathcal{N}} \backslash n} (v_i^{\textrm{opt}})^* \bm{\Phi}_{ij}  v_j^{\textrm{opt}} +  \sum_{i\in \tilde{\mathcal{N}} \backslash n }  2\Re\{(v_i^{\textrm{opt}})^* b_i \} +\hat{\mathbf{h}}_d^H \mathbf{w}\mathbf{w}^H \hat{\mathbf{h}}_d $. 
Similarly, the interference-plus-noise power in \eqref{single_user_problem_d} can be expressed in a quadratic form of $a_n$ as $a_n^2 \mathbf{R}_{nn}\|\mathbf{w}\|^2 +  a_n e_n + f_n$, where $e_n$ and $f_n$ are not related to $a_n$ and their expressions are omitted for brevity. 
As a result, the objective function of problem \eqref{single_user_problem_d} can be rewritten as
\begin{equation} \label{SINR_bound}
\log \left( 1+ \frac{a_n^2 \mathbf{A}_{nn} +  a_n c_n + d_n}{a_n^2 \mathbf{R}_{nn}\|\mathbf{w}\|^2 +  a_n e_n + f_n} \right).
\end{equation}
From \eqref{SINR_bound},  we observe that $a_n=1$ is not necessarily optimal for maximizing the achievable rate, especially when $\mathbf{R}_{nn} \|\mathbf{w}\|^2 \gg \mathbf{A}_{nn}$ and $e_n \gg c_n$, which could be the case when CSI error becomes large.

Note that in \cite{you2019progressive}, a similar problem formulation to \eqref{single_user_problem_d} is considered under the assumption of single-antenna AP and full IRS reflection, i.e., $|v_n|=1,\forall n \in \mathcal{N}$, and a monotonic convergent algorithm is proposed by employing the block coordinate descent (BCD) method with the semidefinite relaxation (SDR) \cite{LuoSDR2010} based initialization. This algorithm is referred to as the SDR-BCD algorithm in the sequel, which can also be applied to solve problem \eqref{single_user_problem_d} by properly modifying the BCD method, i.e., when successively refining the reflection coefficients, one-dimensional search is applied over $\mathcal{F}_d$. The complexity of this algorithm is given by $\mathcal{O}(I_{\textrm{BCD}}2^{Q_{\theta}} 2^{Q_{a}}N^3 + (N+1)^{6.5} + I_{\textrm{r}}N^2)$, where $I_\textrm{BCD}$ denotes the number of iterations needed for convergence of the BCD method and $I_{\textrm{r}}$ is the number of Gaussian randomizations used for SDR \cite{Wang2014}. This complexity is quite high mainly due to the SDR-based optimization required for this algorithm. To reduce complexity, we propose in this section a new algorithm, called the penalized Dinkelbach-BSUM algorithm, to solve problem \eqref{single_user_problem_d} by combining the penalty method, Dinkelbach method and BSUM method, which achieves better performance yet with lower complexity as compared to the SDR-BCD method.

\subsection{Joint Discrete Amplitude and Phase-Shift Control}
First, we consider problem \eqref{single_user_problem_d} in the general case of discrete amplitude/phase shift. It is observed that the transmit power constraint $\|\mathbf{w}\|^2 \leq P$ in problem \eqref{single_user_problem_d} must be satisfied with equality at optimality since otherwise, we can always scale $\mathbf{w}$ properly such that its objective value is increased without violating any constraint. 
Besides, the maximum-ratio transmission (MRT) based solution of $\mathbf{w}$ is optimal in the single-user case, i.e., $\mathbf{w} = \sqrt{P}{(\mathbf{v}^H \hat{\mathbf{H}} + \hat{\mathbf{h}}_d^H)^H}/{\|\mathbf{v}^H \hat{\mathbf{H}} + \hat{\mathbf{h}}_d^H\|}$. Therefore, problem \eqref{single_user_problem_d} is equivalent to 
\begin{equation} \label{single_user_problem_d_eq2}
\begin{aligned}
\max \limits_{\mathbf{v}}\; & \frac{P\|\mathbf{v}^H \hat{\mathbf{H}} + \hat{\mathbf{h}}_d^H\|^2}{ P (\mathbf{v}^H \mathbf{R} \mathbf{v}+\mathbf{v}^H\mathbf{r}+ \mathbf{r}^H\mathbf{v} +\bar{v}_{11} ) + \sigma^2} \\
\textrm{s.t.}\; & v_n \in \mathcal{F}_d,\;\forall n \in \mathcal{N}.
\end{aligned}
\end{equation}

To solve problem \eqref{single_user_problem_d_eq2}, we introduce an auxiliary vector $\mathbf{u} = [u_1, \cdots, u_N]^T$, which satisfies $\mathbf{u} = \mathbf{v}$. This is to facilitate parallel updating of the IRS reflection coefficients in $\mathbf{v}$ and thus simplify their optimization, as will be specified later. Consequently, we have the following equivalent form of problem \eqref{single_user_problem_d_eq2}:
\begin{subequations} \label{single_user_problem_d_eq3}
\begin{align}
 \max \limits_{\mathbf{v},\;\mathbf{u}}\; & \frac{P(\mathbf{v}^H \hat{\mathbf{H}} \hat{\mathbf{H}}^H \mathbf{v}+2\Re\{\mathbf{v}^H \hat{\mathbf{H}} \hat{\mathbf{h}}_d \}+\hat{\mathbf{h}}_d^H \hat{\mathbf{h}}_d )  }{ P (\mathbf{v}^H \mathbf{R} \mathbf{v}+2\Re\{ \mathbf{v}^H\mathbf{r}\}+\bar{v}_{11} ) + \sigma^2 } \label{obj}\\
 \textrm{s.t.}\; & \|\mathbf{v}\|\leq N, \label{bound_constraint} \\
& \mathbf{v} = \mathbf{u}, \label{eq_constraint}\\
& u_n \in \mathcal{F}_d,\;\forall n \in\mathcal{N}, \label{reflection_constraint}
\end{align}
\end{subequations}
where the constraint \eqref{bound_constraint} is added without loss of generality to ensure that the optimization with respect to $\mathbf{v}$ with fixed $\mathbf{u}$ provides a bounded objective value. Note that the equality constraint \eqref{eq_constraint} hinders the alternating optimization of $\mathbf{v}$ and $\mathbf{u}$.
To address this issue, we convert the equality constraint \eqref{eq_constraint} into a quadratic function and then add it as
a penalty term in the denominator of \eqref{obj}, yielding the following optimization problem:
\begin{equation} \label{single_user_problem_d_relax}
	\begin{aligned}
	 \max \limits_{\mathbf{v},\;\mathbf{u}}\; & \frac{P(\mathbf{v}^H \hat{\mathbf{H}} \hat{\mathbf{H}}^H \mathbf{v}+2\Re\{\mathbf{v}^H \hat{\mathbf{H}} \hat{\mathbf{h}}_d \}+\hat{\mathbf{h}}_d^H \hat{\mathbf{h}}_d )  }{ P (\mathbf{v}^H \mathbf{R} \mathbf{v}+2\Re\{ \mathbf{v}^H\mathbf{r}\}+\bar{v}_{11} ) + \sigma^2 +\frac{1}{\beta}\|\mathbf{v} - \mathbf{u}\|^2 } \\
	\textrm{s.t.}\; & \|\mathbf{v}\|\leq N,\\
	& u_n \in \mathcal{F}_d,\;\forall n \in\mathcal{N},
	\end{aligned}
\end{equation}
where $\beta> 0$ denotes the penalty coefficient used for penalizing the violation of the equality constraint \eqref{eq_constraint}.
The proposed penalized Dinkelbach-BSUM algorithm aims to solve problem \eqref{single_user_problem_d_relax} via two loops. In the outer loop, we gradually decrease the value of $\beta$, such that a solution that satisfies the equality constraint \eqref{eq_constraint} within a predefined accuracy can be obtained. While in the inner loop, we apply the Dinkelbach method and BSUM method to iteratively optimize $\mathbf{v}$ and $\mathbf{u}$ with one of them being fixed.
It is worth pointing out that although the equality constraint \eqref{eq_constraint} is relaxed in problem \eqref{single_user_problem_d_relax}, any solution obtained by solving \eqref{single_user_problem_d_relax} always satisfies this equality constraint when $\beta \rightarrow 0$ ($\frac{1}{\beta} \rightarrow \infty$) since otherwise, $\frac{1}{\beta}\|\mathbf{v} - \mathbf{u}\|^2$ will go to infinity (i.e., the objective value will become zero) and we can always let $\mathbf{v}=\mathbf{u}$ to achieve a larger objective value.

Next, by applying the Dinkelbach method in the inner loop, we can transform problem \eqref{single_user_problem_d_relax} into the following Dinkelbach subproblem:
\begin{equation} \label{Dinkelbach_problem}
\begin{aligned}
\max \limits_{{\mathbf{v}},\;\mathbf{u}}\;& {P(\mathbf{v}^H \hat{\mathbf{H}} \hat{\mathbf{H}}^H \mathbf{v}+2\Re\{\mathbf{v}^H \hat{\mathbf{H}} \hat{\mathbf{h}}_d \}+\hat{\mathbf{h}}_d^H \hat{\mathbf{h}}_d )  } \\
 - y& \left({ P   (\mathbf{v}^H \mathbf{R} \mathbf{v}+2\Re\{ \mathbf{v}^H\mathbf{r}\}+\bar{v}_{11} ) + \sigma^2}+\frac{1}{\beta}\|\mathbf{v} - \mathbf{u}\|^2 \right)\\
\textrm{s.t.}\;& \|{\mathbf{v}}\|\leq N, \\
&  u_n \in \mathcal{F}_d, \;\forall n \in\mathcal{N},
\end{aligned}
\end{equation}
where $y$ is the Dinkelbach variable and can be iteratively updated by \eqref{dinkelbach_variable_update} (shown at the top of this page)
\begin{figure*}
\begin{equation} \label{dinkelbach_variable_update}
y[i_d+1] =   \frac{P(\mathbf{v}^H[i_d] \hat{\mathbf{H}} \hat{\mathbf{H}}^H \mathbf{v}[i_d]+2\Re\{\mathbf{v}^H[i_d] \hat{\mathbf{H}} \hat{\mathbf{h}}_d \}+\hat{\mathbf{h}}_d^H \hat{\mathbf{h}}_d )  }{ P (\mathbf{v}^H[i_d] \mathbf{R} \mathbf{v}[i_d]+2\Re\{ \mathbf{v}^H[i_d]\mathbf{r}\}+\bar{v}_{11} ) + \sigma^2 + \frac{1}{\beta}\|\mathbf{v}[i_d] - \mathbf{u}[i_d]\|^2}
\end{equation}
\hrulefill
\end{figure*}
with $i_d$ being the inner iteration index. To tackle problem \eqref{Dinkelbach_problem}, we employ the BSUM method to approximate it as follows (utilizing the first-order Taylor expansion and ignoring constant terms):
\begin{equation} \label{approximate_problem}
\begin{aligned}
\min \limits_{\mathbf{v},\;\mathbf{u}}\;& y{ P (\mathbf{v}^H \mathbf{R} \mathbf{v}+2\Re\{ \mathbf{v}^H\mathbf{r}\})} +\frac{y}{\beta}\|\mathbf{v} - \mathbf{u}\|^2 \\
&- {P( 2\Re\{ (\hat{\mathbf{H}} \hat{\mathbf{H}}^H \mathbf{v}[i_d])^H (\mathbf{v} - \mathbf{v}[i_d]) \}+2\Re\{\mathbf{v}^H \hat{\mathbf{H}} \hat{\mathbf{h}}_d \})}\\
\textrm{s.t.}\; &\|\mathbf{v}\|\leq N,\\
&  u_n \in \mathcal{F}_d,\;\forall n \in\mathcal{N}.
\end{aligned}
\end{equation}
Then, problem \eqref{approximate_problem} can be solved by alternately optimizing two blocks of variables, i.e., $\mathbf{v}$ and $\mathbf{u}$. Specifically, we can obtain the following two subproblems:
\begin{equation} \label{subproblem1}
\begin{aligned}
\min \limits_{\mathbf{v}}\; & y{ P (\mathbf{v}^H \mathbf{R} \mathbf{v}+2\Re\{ \mathbf{v}^H\mathbf{r}\}) } +\frac{y}{\beta}\|\mathbf{v} - \mathbf{u}\|^2 \\
& - {P( 2\Re\{ (\hat{\mathbf{H}} \hat{\mathbf{H}}^H \mathbf{v}[i_d])^H (\mathbf{v} - \mathbf{v}[i_d]) \}+2\Re\{\mathbf{v}^H \hat{\mathbf{H}} \hat{\mathbf{h}}_d \})}\\
\textrm{s.t.}\;& \|\mathbf{v}\|\leq N,
\end{aligned}
\end{equation}
\begin{equation} \label{subproblem2}
\begin{aligned}
\min \limits_{\mathbf{u}}\;& \frac{y}{\beta}\|\mathbf{v} - \mathbf{u}\|^2 \\
\textrm{s.t.}\; & u_n \in \mathcal{F}_d,\;\forall n \in\mathcal{N}.
\end{aligned}
\end{equation}
Problem \eqref{subproblem1} is a convex quadratically constrained quadratic program (QCQP) problem with only one constraint and the Slater's condition holds for it \cite{ConvexOptimization}. Therefore, it can be efficiently solved by applying the Lagrange duality method. Specifically, by letting $\mu$ denote the dual variable and exploiting the first-order optimality condition, we can obtain the following equality:
\begin{equation}
	\begin{aligned}
\Bigg(y P\mathbf{R} + & \left(\frac{y}{\beta}+\mu\right) \mathbf{I}\Bigg)  \bar{\mathbf{v}}\\
& = \frac{y}{\beta} \mathbf{u} + P \hat{\mathbf{H}} \hat{\mathbf{H}}^H \mathbf{v}[i_d] +P \hat{\mathbf{H}} \hat{\mathbf{h}}_d-yP \mathbf{r}.
\end{aligned}
\end{equation}
As s result, the optimal solution to problem \eqref{subproblem1} is given by
\begin{equation} \label{variable_update}
	\begin{aligned}
{\mathbf{v}} (\mu)= & \left(y P\mathbf{R} +  \left(\frac{y}{\beta}+\mu\right) \mathbf{I} \right) ^{-1} \\
&\times \left( \frac{y}{\beta} \mathbf{u} + P\hat{\mathbf{H}} \hat{\mathbf{H}}^H \mathbf{v}[i_d] + P\hat{\mathbf{H}} \hat{\mathbf{h}}_d-yP\mathbf{r} \right),
\end{aligned}
\end{equation}
where if $\|{\mathbf{v}} (0)\|\leq N$, then ${\mathbf{v}} (0)$ is the optimal solution; otherwise, the optimal dual variable $\mu^{\textrm{opt}}$ can be obtained via the bisection method. On the other hand, for problem \eqref{subproblem2}, we note that $\{u_n\}$ are decoupled in both the objective function and constraints. Thus, the optimal solution of  problem \eqref{subproblem2} can be obtained in parallel as follows:
\begin{equation} \label{u_update}
u_n^{\textrm{opt}} = \hat{a}_n  e^{\jmath \angle u_n },
\end{equation}
where $\angle u_n = \arg\min\limits_{\angle u_n \in \mathcal{S}} |\angle u_n - \angle v_n|$ and $\hat{a}_n = \arg\min\limits_{a_n \in \mathcal{A}} |a_n e^{\jmath \angle u_n} -v_n |$.

In summary, we can solve problem \eqref{single_user_problem_d_relax} by iterating over \eqref{dinkelbach_variable_update}, \eqref{variable_update} and \eqref{u_update} in the inner loop and gradually decreasing $\beta$ in the outer loop, and the proposed penalized Dinkelbach-BSUM algorithm is shown in Algorithm \ref{Dinkelbach_BSUM_algorithm_singleuser}. The constraint violation $\| \mathbf{v} - \mathbf{u}\|_{\infty}$ is evaluated as a measure of convergence and the scaling constant $c$ is imposed to gradually decrease the penalty coefficient $\beta$ such that $\| \mathbf{v} - \mathbf{u}\|_{\infty}$ is enforced to zero eventually. 
Note that in Steps 12-14, we further apply the BCD method to successively refine the IRS reflection coefficients. The optimal reflection coefficient for each
element is found by maximizing \eqref{obj} via one-dimensional search over $\mathcal{F}_d$, with those of the others being fixed, until \eqref{obj} converges.
The convergence of the proposed Algorithm \ref{Dinkelbach_BSUM_algorithm_singleuser} is given by the following proposition.
\begin{prop} \label{prop_convergent} 
	\emph{
With any given penalty coefficient $\beta$, the proposed Dinkelbach-BSUM algorithm in Steps 3-9 of Algorithm \ref{Dinkelbach_BSUM_algorithm_singleuser} is monotonically convergent.}
\end{prop}
\begin{proof} \vspace{-0.5em}
	Please refer to Appendix \ref{appendix_C}.
\end{proof}
\noindent Together with the fact that the BCD method is guaranteed to converge, we can conclude that Algorithm \ref{Dinkelbach_BSUM_algorithm_singleuser} is convergent. Compared to the conventional Dinkelbach method, the proposed Algorithm \ref{Dinkelbach_BSUM_algorithm_singleuser} is more general since it can guarantee convergence even when the Dinkelbach subproblem \eqref{Dinkelbach_problem} is not globally solved and also handle the case with discrete variables. Besides, since problem \eqref{single_user_problem_d_eq2} is an integer (nonlinear) program, which is NP-complete, obtaining a convergent solution is perhaps the best we can do for problem \eqref{single_user_problem_d_eq2} with an acceptable computational complexity. Further investigation into more advanced algorithms is still open.

Note that the complexity of the proposed algorithm is mainly due to the update of $\mathbf{v}$ in \eqref{variable_update}, thus can be shown to be $\mathcal{O}(I_{\textrm{BCD}}2^{Q_{\theta}} 2^{Q_{a}}N^3 + I_P I_{D}N^{3}\log(1/\epsilon_{\textrm{bi}}))$, where $\epsilon_{\textrm{bi}}$ is the accuracy of the bisection method, and $I_P$ and $I_D$ denote the number of iterations required by the penalty method and Dinkelbach-BSUM method, respectively.

\begin{remark}
\emph{In Algorithm \ref{Dinkelbach_BSUM_algorithm_singleuser}, the key idea to deal with the discrete variables is by introducing a redundancy copy (i.e., the auxiliary vector $\mathbf{u}$) of the reflection coefficient vector $\mathbf{v}$, where the variables in $\mathbf{v}$ are relaxed to continuous values while those in $\mathbf{u}$ still remain discrete. Then, we introduce a penalty term $\frac{1}{\beta}\|\mathbf{v} - \mathbf{u}\|^2$ and optimize $\mathbf{v}$ and $\mathbf{u}$ in an iterative manner until $\mathbf{v}=\mathbf{u}$ is satisfied up to an acceptable accuracy. Note that this idea is different from the conventional relax-and-then-quantize scheme \cite{Wu2019Discrete}. Besides, although there are more discrete variables in problem \eqref{single_user_problem_d} as compared to the problem with phase-shift control only, this issue has little impact on the complexity of Algorithm \ref{Dinkelbach_BSUM_algorithm_singleuser}, as can be observed from \eqref{variable_update} and \eqref{u_update}. }
\end{remark}

\begin{algorithm}[t] \small
	\caption{{Proposed Algorithm for Solving Problem \eqref{single_user_problem_d_eq3} }} \label{Dinkelbach_BSUM_algorithm_singleuser}
	\begin{algorithmic}[1]
		\STATE Initialize ${\mathbf{v}}[0]$, $\mathbf{u}[0]$, $y[0]$ and $\beta$, set $\epsilon_{d}>0$ and $\epsilon_{p}>0$.
		\REPEAT
		\STATE Let the iteration index $i_{d} \leftarrow 0$.
		\REPEAT
		\STATE Update ${\mathbf{v}}[i_d+1]$, ${\mathbf{u}}[i_d+1]$ and the Dinkelbach variable $y[i_d+1]$ successively according to \eqref{variable_update}, \eqref{u_update} and \eqref{dinkelbach_variable_update}, respectively.
		\STATE Let $i_{d} \leftarrow i_{d} + 1$.
		\UNTIL{The fractional decrease of $y$ is below $\epsilon_{d}$. }
		\STATE Update the penalty coefficient $\beta$ by $\beta \leftarrow c \beta$.
		\UNTIL{The constraint violation $\|\mathbf{v}-\mathbf{u}\|_{\infty} $ is below $\epsilon_{p}$.}
			\REPEAT
		\STATE Update $\{v_n\}$ successively by one-dimensional search over $\mathcal{F}_d$.
		\UNTIL{The fractional increase of \eqref{obj} is below $\epsilon_{d}$. }
	\end{algorithmic}
\end{algorithm}

\subsection{Joint Continuous Amplitude and Phase-Shift Control}
In this subsection, we consider the continuous counterpart of problem \eqref{single_user_problem_d_eq2}, which can be expressed as follows:
\begin{equation} \label{single_user_problem_c_eq}
	\begin{aligned}
\max \limits_{\mathbf{v}}\; & \frac{P\|\mathbf{v}^H \hat{\mathbf{H}} + \hat{\mathbf{h}}_d^H\|^2}{P (\mathbf{v}^H \mathbf{R} \mathbf{v}+2\Re\{ \mathbf{v}^H\mathbf{r}\}+\bar{v}_{11} ) + \sigma^2 } \\
\textrm{s.t.}\; & v_n \in \mathcal{F}_c,\;\forall n \in \mathcal{N}.
\end{aligned}
\end{equation}
It is worth noting that the proposed penalized Dinkelbach-BSUM algorithm (i.e., Algorithm \ref{Dinkelbach_BSUM_algorithm_singleuser}) can be easily modified to solve problem \eqref{single_user_problem_c_eq} and the only difference lies in the $\mathbf{u}$-subproblem, which is shown as follows:
\begin{equation} \label{subproblem2_c}
	\begin{aligned}
\min \limits_{\mathbf{u}}\; & \frac{y}{\beta}\|\mathbf{v} - \mathbf{u}\|^2 \\
\textrm{s.t.}\; & |u_n|\leq 1,\;\forall n \in\mathcal{N}.
\end{aligned}
\end{equation}
Problem \eqref{subproblem2_c} is a simple projection problem with all the variables being decoupled in the objective function and constraints. Thus, the optimal solution can be easily obtained by
\begin{equation} \label{u_update_c}
u_n = \left\{\begin{array}{l}
v_n,\;\textrm{if}\; |v_n| \leq 1,\\
\frac{v_n}{|v_n|}, \;\textrm{otherwise}.
\end{array}
\right.
\end{equation}

Next, we modify the corresponding BCD method to further refine the IRS reflection coefficients. Since one-dimensional search over $|v_n| \leq 1,\;\forall n \in \mathcal{N}$ is inefficient, we apply the Lagrange duality method to obtain the optimal reflection coefficient of each element, by fixing those of the others, which leads to semi-closed-form expressions. Specifically, by letting $\bm{\Phi} = [ \hat{\mathbf{h}}_d^H \hat{\mathbf{h}}_d, \hat{\mathbf{h}}^H \hat{\mathbf{H}}^H; \hat{\mathbf{H}} \hat{\mathbf{h}},  \hat{\mathbf{H}} \hat{\mathbf{H}}^H ]$, we can equivalently express problem \eqref{single_user_problem_c_eq} in a more compact form as
\begin{equation} \label{continuous_problem}
	\begin{aligned}
\max \limits_{\tilde{\mathbf{v}}}\; & \frac{P\tilde{\mathbf{v}}^H \bm{\Phi} \tilde{\mathbf{v}}}{ P  \tilde{\mathbf{v}}^H \bar{\mathbf{V}}\tilde{\mathbf{v}}  + \sigma^2} \\
\textrm{s.t.}\; & |v_n| \leq 1,\;\forall n \in \mathcal{N}.
\end{aligned}
\end{equation}
For brevity, the details of applying the BCD method to solve problem \eqref{continuous_problem} is presented in Appendix \ref{appendix_BCD_continuous}.

To summarize, the proposed algorithm to solve problem \eqref{single_user_problem_c_eq} is given in Algorithm \ref{BCD_algorithm_singleuser}. Similar to the discrete case, the complexity of Algorithm \ref{BCD_algorithm_singleuser} can be shown to be $\mathcal{O}(I_{\textrm{BCD}}N^3 + I_P I_{D}N^{3}\log(1/\epsilon_{\textrm{bi}}))$. Besides, we can observe that by combing the techniques used in Algorithm \ref{Dinkelbach_BSUM_algorithm_singleuser} and Algorithm \ref{BCD_algorithm_singleuser}, the proposed algorithm can be easily extended to the other cases, e.g., continuous amplitude/discrete phase shift (CADP) and discrete amplitude/continuous phase shift (DACP). The details are omitted for brevity.

\begin{algorithm}[tbp] \small
	\caption{{Proposed Algorithm for Solving Problem \eqref{single_user_problem_c_eq} }} \label{BCD_algorithm_singleuser}
	\begin{algorithmic}[1]
		\STATE Initialize ${\mathbf{v}}$ by running Algorithm \ref{Dinkelbach_BSUM_algorithm_singleuser} with \eqref{u_update} replaced by \eqref{u_update_c}, set $\epsilon_{c}>0$.
		\REPEAT
		\STATE Update $\{v_n\}$ successively by solving problem \eqref{n_subproblem}.
		\UNTIL{The fractional increase of the objective function of problem \eqref{continuous_problem} is below $\epsilon_{c}$. }
	\end{algorithmic}
\end{algorithm}

\section{Multiuser System} \label{Section_MU}
In this section, we solve problem \eqref{robust_problem_d} under the general multiuser setup. In this case, the proposed algorithms in Section \ref{Section_SU} cannot be applied since the objective function is not in a fractional form as given in \eqref{single_user_problem_d_eq2}. To tackle this problem, we first employ the weighted sum mean squared error minimization (WMMSE) method \cite{Shi2011WMMSE} to transform problem \eqref{robust_problem_d} into a more tractable form. Then, we leverage the PDD framework to propose a double-loop iterative algorithm, where the inner loop seeks to solve an augmented Lagrangian (AL) problem using a block minimization technique, while the outer loop updates the dual variables and penalty coefficient according to the constraint violation, until the convergence is achieved. Interestingly, we show that each block of variables (including the active precoders $\{ \mathbf{w}_k\}$ and IRS reflection coefficients $\mathbf{v}$) in the proposed algorithm can be updated either in closed-form or by the simple bisection method (similar to the single-user case). Besides, the proposed PDD-based algorithm can be easily modified to address the continuous/discrete amplitude and/or phase-shift cases.

Also note that the rationale given in the single-user case for why amplitude control is generally required is also valid for the multiuser case, i.e., problems \eqref{robust_problem_d} and \eqref{robust_problem_c}. This is because for the $n$-th reflecting element, its resultant interference-plus-noise power in the multiuser case can also be written as a quadratic function of $a_n$, and the achievable rate of each user can be expressed in a similar form as that in \eqref{SINR_bound}. Intuitively, since there is more severe multiuser interference in this case, the performance gain offered by amplitude control is expected to be larger as compared to the single-user case (as will be shown by simulation in Section \ref{Section_Simulation}).

\subsection{Joint Discrete Amplitude and Phase-Shift Optimization} 

First, by letting $\hat{\mathbf{h}}_{k} = \hat{\mathbf{H}}_{k}^H \mathbf{v} + \hat{\mathbf{h}}_{d,k}$ denote the estimated effective channel vector between  the AP and user $k$, problem \eqref{robust_problem_d} can be written in a more compact form as
\begin{equation} \label{robust_problem_equi_x}
\begin{aligned}
\max \limits_{\{\mathbf{w}_k\}, \mathbf{v}} \; & \sum\limits_{k \in \mathcal{K}} \alpha_k r_k(\hat{\mathbf{h}}_{k},\{\mathbf{w}_k\}, \mathbf{v}) \\
\textrm{s.t.} \; & \sum\limits_{k \in \mathcal{K}}\|\mathbf{w}_k\|^2 \leq P,\\
& v_n \in \mathcal{F}_d,\;\forall n \in \mathcal{N},
\end{aligned}
\end{equation}
where $r_k(\hat{\mathbf{h}}_{k},\{\mathbf{w}_k\}, \mathbf{v})$ is given in \eqref{expression_rk}, shown at the top of this page.
\begin{figure*}
\begin{equation} \label{expression_rk}
r_k(\hat{\mathbf{h}}_{k},\{\mathbf{w}_k\}, \mathbf{v}) = \log \left(1+\frac{|\hat{\mathbf{h}}_{k}^H \mathbf{w}_k|^2}{\sum\limits_{j \in \mathcal{K}\backslash k}|\hat{\mathbf{h}}_{k}^H \mathbf{w}_j|^2 +(\bar{v}_{11,k}+\mathbf{v}^H\mathbf{r}_k+ \mathbf{r}_k^H\mathbf{v} + \mathbf{v}^H \mathbf{R}_k \mathbf{v})\sum\limits_{j\in\mathcal{K}} \|\mathbf{w}_j\|^2 + \sigma_k^2}\right)
\end{equation} 
\hrulefill
\end{figure*}
Then, to transform problem \eqref{robust_problem_equi_x} into a more tractable form, we apply the WMMSE method that leads to the following proposition.
\begin{prop} 
	\emph{
	 Problem \eqref{robust_problem_equi_x} has the same globally optimal solution as the following WMMSE problem:
	\begin{equation} \label{robust_problem_interference_equi2_x}
	\begin{aligned}
	\min \limits_{\{\mathbf{w}_k,q_k,g_k \}, \mathbf{v}} \; & \sum\limits_{k \in \mathcal{K}} \alpha_k(q_k e_k - \log q_k) \\
	\textrm{s.t.} \; & \sum\limits_{k\in\mathcal{K}} \|\mathbf{w}_k\|^2 \leq P,\\
	& v_n \in \mathcal{F}_d,\;\forall n \in \mathcal{N},
	\end{aligned}
	\end{equation}
	where $e_k= \mathbb{E}\{ (g_k^H y_k - s_k) (g_k^H y_k - s_k)^H\} $ denotes the mean squared error (MSE) of user $k$ and is given by
	\begin{equation} 
	\begin{aligned}
	e_k = & |g_k|^2\Big(\sum\limits_{j \in\mathcal{K}} |\hat{\mathbf{h}}_{k}^H \mathbf{w}_j|^2+ (\bar{v}_{11,k}+\mathbf{v}^H\mathbf{r}_k+ \mathbf{r}_k^H\mathbf{v} \\
	&+ \mathbf{v}^H \mathbf{R}_k \mathbf{v}) \sum\limits_{j\in \mathcal{K}} \| \mathbf{w}_j\|^2+ \sigma_k^2 \Big)
	 -2\Re(g_k^H \hat{\mathbf{h}}_{k}^H \mathbf{w}_k) + 1,
	\end{aligned}
	\end{equation}
$g_k$ and $q_k$ denote the receive (scaling) coefficient and weighting factor for user $k$, respectively.}
\end{prop}
\begin{proof} \vspace{-0.5em}
	Since $\{g_k\}$ and $\{q_k\}$ only appear in the objective function of problem \eqref{robust_problem_interference_equi2_x}, by substituting their optimal solutions into the objective function of problem \eqref{robust_problem_interference_equi2_x} and following \cite[Theorem 3]{Shi2011WMMSE}, the equivalence between the two problems can be established.
\end{proof}

In order to simplify the optimization of $\mathbf{v}$ and facilitate its parallel updating in the proposed PDD-based algorithm, we introduce an auxiliary variable $\mathbf{u}$ (similar to the single-user case) and problem \eqref{robust_problem_interference_equi2_x} can then be equivalently transformed to 
\begin{equation} \label{robust_problem_interference_equi_x}
\begin{aligned}
\min \limits_{\{\mathbf{w}_k,q_k,g_k \}, \mathbf{v},\mathbf{u}} \; & \sum\limits_{k \in \mathcal{K}} \alpha_k(q_k e_k - \log q_k) \\
\textrm{s.t.} \; & \sum\limits_{k \in \mathcal{K}}\|\mathbf{w}_k\|^2 \leq P,\\
& \mathbf{v} = \mathbf{u},\\
& u_n \in \mathcal{F}_d,\;\forall n \in \mathcal{N}.
\end{aligned}
\end{equation}

Next, in the inner loop of the proposed PDD-based algorithm, we solve the following AL problem of \eqref{robust_problem_interference_equi_x} by applying the BCD method (see \cite{Zhao2019CL} for the derivation of the AL term):
\begin{equation} \label{AL_problem}
\begin{aligned}
\min \limits_{\{\mathbf{w}_k,g_k,q_k\},\mathbf{v}, \mathbf{u}} \; & \sum\limits_{k\in\mathcal{K}} \alpha_k(q_k e_k - \log q_k)  + \frac{1}{2\beta} \|\mathbf{v}-\mathbf{u}+\beta \bm{\mu} \|^2\\
\textrm{s.t.}\; & \sum\limits_{k\in\mathcal{K}}\|\mathbf{w}_k\|^2 \leq P,\\
& u_n \in \mathcal{F}_d,\;\forall n \in \mathcal{N}.
\end{aligned}
\end{equation}
where $\bm{\mu}$ is the dual variable vector associated with the constraint  $\mathbf{v} = \mathbf{u}$ and $\beta$ is the penalty coefficient. In particular, we observe that by dividing the optimization variables into the following five blocks: $\{\mathbf{w}_k\}$, $\{g_k\}$, $\{q_k\}$, $\mathbf{v}$ and $\mathbf{u}$, the optimization of each block of variables with the others fixed is much simplified as compared with those in problem \eqref{robust_problem_d}. As a result, we can successively solve each of the subproblems and the details are given as follows.

\subsubsection{The optimization of $g_k$} By fixing all the other variables, minimizing the (weighted) sum MSE $\sum_{k\in\mathcal{K}} \alpha_k q_k e_k $ leads to the linear minimum MSE (MMSE) receive coefficient, which is given by \eqref{LMMSE} shown at the top of the next page.
\begin{figure*}
\begin{equation} \label{LMMSE}
	g_k = \frac{\hat{\mathbf{h}}_{k}^H \mathbf{w}_k}{\sum\limits_{j \in\mathcal{K}} |\hat{\mathbf{h}}_{k}^H \mathbf{w}_j|^2+ (\bar{v}_{11,k}+\mathbf{v}^H\mathbf{r}_k+ \mathbf{r}_k^H\mathbf{v} + \mathbf{v}^H \mathbf{R}_k \mathbf{v}) \sum\limits_{j\in\mathcal{K}} \|\mathbf{w}_j\|^2+ \sigma_k^2 }
\end{equation}
\hrulefill
\end{figure*}
\subsubsection{The optimization of $q_k$} The optimal solution can be easily obtained as $q_k=\frac{1}{e_k}$.

\subsubsection{The optimization of $\mathbf{w}_k$} The update of the transmit precoders $\{\mathbf{w}_k \}$ can be conducted by solving the following QCQP problem:
\begin{equation} \label{w_k_problem}
	\begin{aligned}
\min \limits_{\{\mathbf{w}_k\}} \; & \sum\limits_{k\in\mathcal{K}} \alpha_kq_k e_k \\
\textrm{s.t.}\; & \sum\limits_{k\in\mathcal{K}}\|\mathbf{w}_k\|^2 \leq P.
\end{aligned}
\end{equation}
Since there is only one constraint in \eqref{w_k_problem} and $\mathbf{w}_k$'s are decoupled in the Lagrangian function associated with problem \eqref{w_k_problem}, its optimal solution can be obtained by exploiting the first-order optimality condition as follows:
\begin{equation} \label{w_k_update}
\begin{aligned}
\mathbf{w}_k(\nu) = &  \alpha_k q_k \Big(\sum\limits_{j \in \mathcal{K}}  (\alpha_j q_j  |g_j|^2\hat{\mathbf{h}}_{j}  \hat{\mathbf{h}}_{j}^H+\alpha_j q_j |g_j|^2  (\bar{v}_{11,k}\\
& +\mathbf{v}^H\mathbf{r}_k+ \mathbf{r}_k^H\mathbf{v} + \mathbf{v}^H \mathbf{R}_k \mathbf{v}) \mathbf{I})+ \nu \mathbf{I}\Big)^{-1} \hat{\mathbf{h}}_{k} g_k,
\end{aligned}
\end{equation}
where $\nu$ denotes the dual variable associated with the transmit power constraint. Similar to problem \eqref{subproblem1}, if $\sum_{k\in \mathcal{K}}\|\mathbf{w}_k(0)\|^2 \leq P $, then $\{\mathbf{w}_k(0)\}$ is the optimal solution; otherwise, we can always find the optimal dual variable via the bisection method.

\subsubsection{The optimization of $\mathbf{v}$} It can be observed that the $\mathbf{v}$-subproblem becomes the following unconstrained quadratic programming (QP) problem:
\begin{equation}
\begin{aligned}
\min\limits_{\mathbf{v}}\; & -\sum\limits_{k \in \mathcal{K}} \alpha_k q_k 2\Re(g_k^H ( \mathbf{v}^H \hat{\mathbf{H}}_{k} + \hat{\mathbf{h}}_{d,k}^H) \mathbf{w}_k) \\
& + \sum\limits_{k \in \mathcal{K}} \alpha_k q_k  |g_k|^2\sum\limits_{j \in\mathcal{K}} | ( \mathbf{v}^H \hat{\mathbf{H}}_{k} + \hat{\mathbf{h}}_{d,k}^H) \mathbf{w}_j|^2 \\
& + \sum\limits_{k \in \mathcal{K}} \alpha_k q_k   |g_k|^2 (\mathbf{v}^H\mathbf{r}_k+ \mathbf{r}_k^H\mathbf{v} + \mathbf{v}^H \mathbf{R}_k \mathbf{v})   \sum\limits_{j \in \mathcal{K}} \|\mathbf{w}_j\|^2 \\
& + \frac{1}{2\beta} \|\mathbf{v}-\mathbf{u}+\beta \bm{\mu} \|^2,
\end{aligned}
\end{equation}
whose optimal solution can be expressed as $\mathbf{v} = \bar{\mathbf{C}}^{-1} \bar{\mathbf{d}}$, where 
\begin{equation}
\begin{aligned}
\bar{\mathbf{C}} = &\sum\limits_{k \in \mathcal{K}} \alpha_k q_k   |g_k|^2  \sum\limits_{j\in\mathcal{K}} \hat{\mathbf{H}}_{k} \mathbf{w}_j \mathbf{w}_j^H \hat{\mathbf{H}}_{k}^H   \\
&+ \sum\limits_{k \in \mathcal{K}} \alpha_k q_k   |g_k|^2 \sum\limits_{j\in\mathcal{K}}\|\mathbf{w}_j\|^2 \mathbf{R}_k + \frac{1}{2\beta}\mathbf{I},
\end{aligned}
\end{equation}
\begin{equation}
\begin{aligned}
\bar{\mathbf{d}} =  & -\sum\limits_{k \in \mathcal{K}} \alpha_k q_k   |g_k|^2  \sum\limits_{j\in\mathcal{K}}\hat{\mathbf{H}}_{k} \mathbf{w}_{j} \mathbf{w}_j^H \hat{\mathbf{h}}_{d,k} +  \frac{1}{2\beta} (\mathbf{u}-\beta \bm{\mu})  \\
&+ \sum\limits_{k \in \mathcal{K}} \alpha_k q_k g_k^H \hat{\mathbf{H}}_{k} \mathbf{w}_k - \sum\limits_{k \in \mathcal{K}} \alpha_k q_k   |g_k|^2 \sum\limits_{j \in \mathcal{K}} \|\mathbf{w}_j \|^2 \mathbf{r}_k.
\end{aligned}
\end{equation}

\subsubsection{The optimization of $\mathbf{u}$} Finally, the $\mathbf{u}$-subproblem is given by
\begin{equation} \label{u_subproblem_interference_x}
	\begin{aligned}
\min \limits_{\mathbf{u}}\; & \|\mathbf{v}-\mathbf{u}+\beta \bm{\mu}\|^2 \\
\textrm{s.t.}\; & u_n \in \mathcal{F}_d,\;\forall n \in \mathcal{N}.
\end{aligned}
\end{equation}
Similar to problem \eqref{subproblem2}, the optimal solution of problem \eqref{u_subproblem_interference_x} can be easily obtained in parallel. The expression is omitted for brevity.

To summarize, the above five updating steps are successively performed in each iteration of the inner BCD method.  In the outer loop, the dual variable is updated by 
\begin{equation} \label{dual_update_multiuser}
\begin{array}{l}
\bm{\mu} = \bm{\mu} + \frac{1}{\beta}(\mathbf{v}-\mathbf{u}).
\end{array}
\end{equation}
The proposed PDD-based algorithm is summarized in Algorithm \ref{PDD_algorithm_multiuser} and it is guaranteed to converge \cite{zhao2019intelligent}. It is noteworthy that Algorithm \ref{PDD_algorithm_multiuser} can also converge without the dual variable $\bm{\mu}$, i.e., let $\bm{\mu}=\mathbf{0}$ (similar to the single-user case). However, by introducing $\bm{\mu}$, Algorithm \ref{PDD_algorithm_multiuser} exhibits better convergence behavior in terms of both the objective value and constraint violation \cite{shi2017penalty}. The complexity of Algorithm \ref{PDD_algorithm_multiuser} is dominated by the matrix inversion operations for optimizing $\{\mathbf{w}_k \}$ and $\mathbf{v}$. Therefore, the overall complexity can be shown to be $\mathcal{O}(I_oI_i(KN^3 \log(1/\epsilon_{\textrm{bi}})+N^3))$, where $I_o$ and $I_i$ denote
respectively the outer and inner iteration numbers required for convergence.

\begin{algorithm}[t] \small
	\caption{Proposed PDD-based Algorithm for Solving Problem \eqref{robust_problem_d}} \label{PDD_algorithm_multiuser}
	\begin{algorithmic}[1]
		\STATE Initialize $\mathbf{v}^0$, $\{\mathbf{w}_k^0\}$, set the outer iteration index $i_{\textrm{out}}=0$. Set $\epsilon_{\textrm{in}}>0$, $\epsilon_{\textrm{out}}>0$ and $c<1$.
		\REPEAT
		\STATE Set the inner iteration index $i_{\textrm{in}} = 0$.
		\REPEAT
		\STATE Update $\{g_k\}$, $\{q_k\}$, $\{\mathbf{w}_k\}$, $\mathbf{v}$ and $\mathbf{u}$ successively.
		\STATE Update the inner teration index: $i_{\textrm{in}} \leftarrow i_{\textrm{in}} + 1$.
		\UNTIL{The fractional decrease of the objective value of \eqref{AL_problem} is below the threshold $\epsilon_{\textrm{in}}$ or the maximum inner iteration number is reached.}
		\STATE Update the dual variables by \eqref{dual_update_multiuser} and decrease the penalty coefficient by $\beta  \leftarrow c\beta $.
		\STATE Update the outer iteration index: $i_{\textrm{out}} \leftarrow i_{\textrm{out}} + 1$.
		\UNTIL{The constraint violation $\| \mathbf{v} - \mathbf{u}\|_{\infty}$ is below the threshold $\epsilon_{\textrm{out}}$.}
	\end{algorithmic}
\end{algorithm}

\subsection{Joint Continuous Amplitude and Phase-Shift Optimization}
In this subsection, we address the continuous counterpart of problem \eqref{robust_problem_d}, by utilizing the PDD framework. It can be readily seen that  the updating steps of $\{g_k,q_k,\mathbf{w}_k \}$ and $\mathbf{v}$ are almost identical to those in the discrete amplitude/phase-shift case. The only difference lies in the optimization of $\mathbf{u}$, which can be decomposed into $N$ subproblems, i.e., 
\begin{equation} \label{subproblem3_c}
	\begin{aligned}
\min \limits_{u_n}\; & |v_n + \beta \mu_n-u_n|^2 \\
\textrm{s.t.}\; & u_n \in \mathcal{F}_c.
\end{aligned}
\end{equation}
Similar to problem \eqref{subproblem2_c}, these subproblems can be solved in parallel and their optimal solutions can be easily obtained. 
Note that for the DACP and CADP cases, the optimization of $\mathbf{u}$ is also the only difference in the proposed PDD-based algorithm and it can be similarly addressed.
Besides, the complexity of the proposed algorithm in the continuous amplitude/phase-shift case is similar to that in the discrete case, since the complexity of optimizing $\{\mathbf{w}_k \}$ and $\mathbf{v}$ dominates.

\begin{remark}
	\emph{
		Note that the proposed PDD-based algorithm (i.e., Algorithm \ref{PDD_algorithm_multiuser}) can be easily modified to solve the single-user problems \eqref{single_user_problem_d} and \eqref{single_user_problem_c_eq}. However, the proposed Algorithm \ref{Dinkelbach_BSUM_algorithm_singleuser} and Algorithm \ref{BCD_algorithm_singleuser} are more suitable for the single-user case since the WMMSE method is not required and thus less auxiliary variables are introduced (or less number of blocks in the inner loop required), which leads to faster convergence and thus more efficient solutions. }
\end{remark}

\section{Simulation Results} \label{Section_Simulation}
In this section, we provide numerical results by simulations to evaluate the performance of the proposed algorithms and draw useful insights. The distance-dependent path loss is modeled as $L = C_0\left({d_{\textrm{link}}}/{D_0}\right)^{-\alpha}$, where $C_0$ is the path loss at the reference distance $D_0 = 1$ meter (m), $d_{\textrm{link}}$ represents the individual link distance and $\alpha$ denotes the path-loss exponent. The path-loss exponents of the AP-user, AP-IRS and IRS-user links are denoted by $\alpha_{Au}$, $\alpha_{AI}$ and $\alpha_{Iu}$, respectively. We assume that the IRS is deployed to serve the users that suffer from severe signal attenuation in the AP-user direct link and thus we set $\alpha_{Au} = 3.6$ and $\alpha_{AI} = \alpha_{Iu} = 2.2$, i.e., the path-loss exponent of the AP-user link is larger than those of the AP-IRS and IRS-user links. In our simulations, a three-dimensional coordinate system is considered where the AP (equipped with a uniform linear array (ULA)) and the IRS (equipped with a uniform rectangular array (UPA)) are located on the $x$-axis and $y$-$z$ plane, respectively. We set $N = N_yN_z$ where $N_y$ and $N_z$ denote the numbers of reflecting elements along the $y$-axis and $z$-axis, respectively. For the purpose of exposition, we fix $N_y = 4$. As shown in Fig. \ref{user_setup}, the reference antenna/element at the AP/IRS are located at $(2\;\textrm{m}, 0, 0)$ and $(0, d_0=45\;\textrm{m}, 2\;\textrm{m})$ and the locations of the users are randomly generated in the cluster. 
To account for small-scale fading, we assume the Rician fading channel model for all channels involved in general. Thus, the AP-IRS channel $\mathbf{G}$ is given by $\mathbf{G} = \sqrt{{\beta_{AI}}/{(1+\beta_{AI})}} \mathbf{G}^{\textrm{LoS}} + \sqrt{{1}/{(1+\beta_{AI})}} \mathbf{G}^{\textrm{NLoS}}$, 
where $\beta_{AI}$ is the Rician factor, $ \mathbf{G}^{\textrm{LoS}}$ and $ \mathbf{G}^{\textrm{NLoS}}$ represent the deterministic line-of-sight (LoS) and Rayleigh fading non-LoS (NLoS) components, respectively. The AP-user and IRS-user channels are also generated by following the similar procedure and the Rician factors of these two links are denoted by $\beta_{Au}$ and $\beta_{Iu}$, respectively. Without loss
of generality, we assume that all users use the same uplink training power during channel training, i.e., $p_{u,k} = p_u$, $\forall k \in \mathcal{K}$, and the user rate weights are the same and thus set to one, i.e.,  $\alpha_k =1$, $\forall k \in \mathcal{K}$. Other system parameters are set as follows unless otherwise specified: $N_r=N+1$, $\sigma_k^2 = -80$ dBm, $C_0 = -30$ dB, $P=26$ dBm, $M=4$,  $\beta_{Au} = \beta_{Iu}   = 0$ and $ \beta_{AI} = 3$ dB. The simulations are implemented in MATLAB R2016a and carried out on a PC with Intel i5 CPU running at 2.3 GHz and with 8 GB RAM.

\begin{figure}[!hhh] 
	\centering
	\scalebox{0.45}{\includegraphics{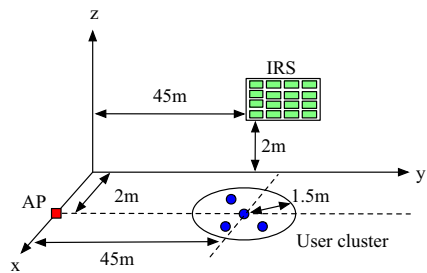}}
	\caption{Simulation setup of the considered IRS-aided multiuser MISO downlink system.}
	\label{user_setup}
\end{figure}

Before rate performance comparison, we first illustrate in Fig. \ref{figure_channel_estimation} the performance of the studied channel estimation method in the single-user case (while the results in the multiuser case are similar and thus omitted due to space limitation), where the normalized MSE (defined as ${\| \bar{\mathbf{H}} - \tilde{\mathbf{H}}\|^2}/{\|\tilde{\mathbf{H}}\|^2 }$) is adopted as the performance metric. In Fig. \ref{figure_channel_estimation} (a), we plot the average normalized MSE versus the uplink training power $p_u$, where $N$ is fixed as $120$. It is observed that the channel estimation performance improves with the increasing of $p_u$ and $Q_{\theta}$, which is expected since larger $p_u$ implies higher channel training signal-to-noise ratio (SNR) and with larger $Q_{\theta}$, we are able to design reflection patterns which are more near-orthogonal. Besides, in Fig. \ref{figure_channel_estimation} (b), we plot the average normalized MSE versus the number of reflecting elements, $N$, with fixed $p_u = 18$ dBm. It can be seen that larger $N$ generally leads to smaller normalized MSE, which is mainly due to the fact that larger $N$ leads to higher aperture gain and thus stronger signal power during channel estimation.\footnote{Note that although the number of unknown channel coefficients increases with $N$, which may deteriorate the channel estimation performance, the number of time slots $N_r$ that are dedicated for channel estimation also increases and this can compensate the performance loss caused by the increase in the number of unknown channel coefficients.} However, the case of $Q_{\theta}=1$ and $N=60$ is especially worth-noting as the associated channel estimation performance is even better than that when $Q_{\theta}=1$ and $N=80$. This is because the reflection patterns are chosen to be the columns of the THM when $Q_{\theta}=1$ \cite{you2019progressive}, and the columns of the $61\times61$ THM (truncated from a $64\times64$ Hadamard matrix) are more near-orthogonal than those of the $81\times81$ THM (truncated from a $96\times96$ Hadamard matrix).

\begin{figure}[!hhh] 
	\centering
	\scalebox{0.35}{\includegraphics{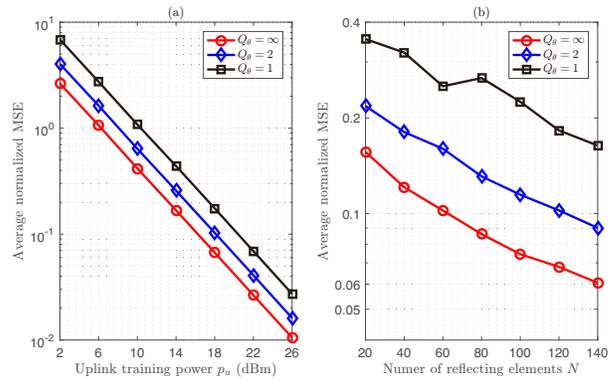}}
	\caption{Average normalized MSE versus uplink training power, $p_u$, or number of reflecting elements, $N$, with different values of $Q_{\theta}$.} 
	\label{figure_channel_estimation}
\end{figure}

\subsection{Single-User System}
\subsubsection{Performance Comparison with Existing Algorithms} 
Then, we investigate the performance of the proposed penalized Dinkelbach-BSUM algorithm  (i.e., Algorithm \ref{Dinkelbach_BSUM_algorithm_singleuser}) with discrete amplitude/ phase-shift control at the IRS. For comparison, we consider two benchmark algorithms: 1) the SDR-BCD algorithm in \cite{you2019progressive} with $100$ randomizations, and 2) the BCD algorithm, where the IRS reflection coefficients are randomly initialized and then the BCD method is employ to successively optimize the coefficient of each element with the others fixed. By varying the value of $p_u$, we examine in Fig. \ref{figure_single_user} the achievable rates obtained by the considered algorithms with $N=120$. It is observed that Algorithm \ref{Dinkelbach_BSUM_algorithm_singleuser} achieves the best performance especially when $p_u \leq 18$ dBm. This is because amplitude control is more important in the low $p_u$ regime; while compared with the benchmark algorithms, Algorithm \ref{Dinkelbach_BSUM_algorithm_singleuser} can directly handle the discrete constraints $v_n \in \mathcal{F}_d,\;\forall n \in \mathcal{N}$ by the introduction of the auxiliary variable $\mathbf{u}$ and penalty term $\frac{1}{\beta}\|\mathbf{v}- \mathbf{u}\|^2$. As a result, it potentially offers a better initial point which is more favorable for the subsequent BCD method. Besides, we compare in Table \ref{Table_runing_time} the running time of the proposed algorithm with that of the SDR-BCD algorithm for various values of $N$. The involved semidefinite programming (SDP) problem in the SDR-BCD algorithm is solved by CVX \cite{CVX}. One can observe that the time consumed by the proposed algorithm is significantly less than that of the SDR-BCD algorithm. Note that the proposed algorithm (as well as the reflection amplitudes and phase shifts) needs to be executed (updated) based on the channel coherence time, which is typically on the order of millisecond (ms). Although the required computational time of the proposed algorithm is larger than the typical channel coherence time, it can be effectively reduced by customizing the algorithm for hardware implementation in practical communication systems.

\begin{figure}[!hhh] 
	\centering
	\scalebox{0.42}{\includegraphics{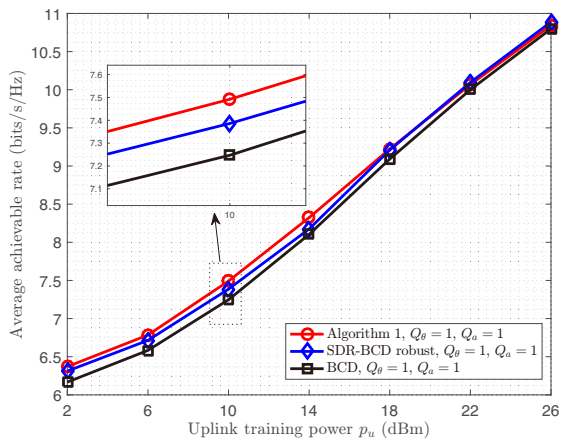}}
	\caption{Achievable rate versus uplink training power, $p_u$.} 
	\label{figure_single_user}
\end{figure}

\renewcommand{\arraystretch}{1.3}
\begin{table}[htbp] 
	\centering
	\caption{Computational Time Comparison} \label{Table_runing_time}
\begin{tabular}{|c|c|c|c|c|}
	\hline  \small
	\multirow{2}{*}{$N$}  & \multicolumn{4}{c|}{Running time (s)} \\ \cline{2-5} 
	& 10      & 20      & 40      & 80      \\ \hline
	Algorithm \ref{Dinkelbach_BSUM_algorithm_singleuser} & 0.0190  & 0.0254  & 0.0454  & 0.1595  \\ \hline
	SDR-BCD            & 0.3429  & 0.4148  & 0.4748  & 1.0260  \\ \hline
\end{tabular}
\end{table}

\subsubsection{Impact of Uplink Training Power, $p_u$} 
In Fig. \ref{Fig_SU_compare_ph}, we compare the achievable rate and IRS elements off percentage (EOP) of the proposed algorithm and the simulation parameters are the same to those in Fig. \ref{figure_channel_estimation} (a). Note that ``$Q_{\theta}=0$, $Q_a = 1$'' denotes the case where the reflection phase shifts are all set to zero and fixed\footnote{For fair comparison, we assume discrete phase shifts with $Q_{\theta} = 1$ in this case for channel training although $Q_{\theta} = 0$ is adopted for data transmission.} while the reflection amplitudes are optimized over $\{0,1\}$, which is thus referred to as \emph{amplitude beamforming}. In contrast, the case ``$Q_{\theta}=1$, $Q_a = 0$'' without amplitude control is referred to as \emph{phase beamforming}. Note that although both amplitude beamforming and phase beamforming contain two reflection states and both cases can be solved by the same algorithm (i.e., Algorithm \ref{Dinkelbach_BSUM_algorithm_singleuser} for the single-user case and Algorithm \ref{PDD_algorithm_multiuser} for the multiuser case), they are different in terms of performance as the former can turn the reflecting elements on or off such that the constructive signal is preserved and the destructive signal can be discarded, while the latter is able to adjust the directions of the reflected signals to align with the desired signal \cite{Wu2020tutorial}. First, it is observed from Fig. \ref{Fig_SU_compare_ph} (a) that the proposed scheme with amplitude control outperforms that with full reflection in both continuous and discrete phase-shift cases. For example, by using $1$-bit phase shifters with on/off amplitude control ($Q_\theta = 1$, $Q_a=1$), the achievable rate is noticeably higher than that achieved by using phase-shift control only ($Q_\theta = 1$, $Q_a = 0$). Second, the performance gain offered by amplitude control is more significant when $p_u$ is lower. This is because amplitude control is more helpful when the CSI is less accurate. In other words, more elements should be turned off to minimize the interference caused by imperfect CSI when the training SNR is lower, as shown in Fig. \ref{Fig_SU_compare_ph} (b). Due to a similar reason, the performance of amplitude beamforming ($Q_\theta= 0$ , $Q_a=1$) is close to that of phase beamforming ($Q_\theta=1$, $Q_a=0$) when $p_u$ is small, although in practice, the former is generally of lower cost to implement as compared to the latter. On the contrary, when $p_u$ is large, phase beamforming is better since the advantage of amplitude control is less significant.

\begin{figure}[!hhh]  
	\centering
	\scalebox{0.35}{\includegraphics{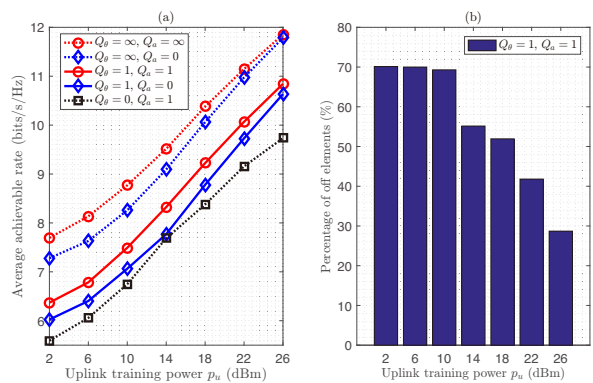}}
	\caption{Achievable rate and IRS EOP versus uplink training power, $p_u$.}
	\label{Fig_SU_compare_ph}
\end{figure}

Moreover, we consider the following two schemes: 1) the conventional scheme by using the MRT beamforming at the AP without the IRS, and 2) the nonrobust scheme by employing the PDD-based algorithm in \cite{zhao2019intelligent}, but ignoring the CSI errors.\footnote{For the nonrobust scheme, the optimization is first carried out by assuming that there is no CSI error in the estimated channel and then the optimized solution is substituted into \eqref{obj} to obtain the achievable rate with imperfect CSI.} 
In Fig. \ref{Fig_SU_compare_benchmark}, we plot the achievable rates by different schemes under the same simulation setup as for Fig. \ref{Fig_SU_compare_ph}.  First, as expected, it is observed that the rate by the proposed scheme is significantly higher than those by the scheme without IRS and the nonrobust scheme. Second, we observe that the performance of the nonrobust scheme is worse than that of the scheme without IRS when $p_u$ is lower than about $20$ dBm. This is due to the fact that the impact of CSI errors is more significant on the AP-IRS-user link than that on the direct AP-user link (since the former experiences ``distance-product'' power loss), therefore ignoring the CSI errors in the nonrobust scheme results in substantial performance degradation. This result also illustrates the importance of robust designs in IRS-aided communication systems by taking into account the CSI errors.
	
	\begin{figure}[!hhh] 
		\centering
		\scalebox{0.42}{\includegraphics{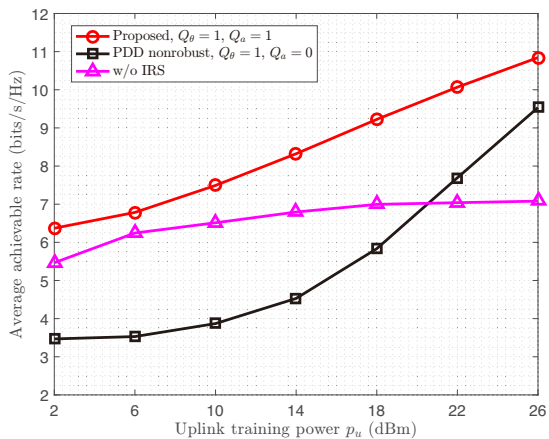}}
		\caption{Achievable rate versus uplink training power, $p_u$.}
		\label{Fig_SU_compare_benchmark}
	\end{figure}

\subsubsection{Impact of Uplink Training Duration, $N_r$} 
Fig. \ref{figure_compare_training_time} plots the achievable rate of the proposed scheme versus $N_r$ with different values of $Q_\theta$ that controls the number of discrete phase-shift levels of IRS, with $N=20$ and $p_u=10$ dBm. In this case, we assume that there are in total $T_0$ symbols in each transmission frame, among which $N_r$ symbols are used for channel training. Consequently, the achievable rate is multiplied by a factor of $\frac{T_0-N_r}{T_0}$ to account for the training overhead. First, it is observed that there exists a tradeoff between the achievable rate and $N_r$. This is due to the fact that smaller $N_r$ leads to coarser CSI, which results in inefficient reflection design and thus less passive beamforming gain, while larger $N_r$ results in less number of symbols for data transmission in each frame. As a result, the optimal value of $N_r$ decreases with the increasing of $Q_\theta$. Second, similar to the results in Fig. \ref{Fig_SU_compare_ph}, the performance gain offered by amplitude control gradually decreases with $N_r$, i.e., when the training SNR increases. In addition, one can observe that the performance gain of amplitude control decreases with the increasing of $Q_\theta$, which is expected since larger $Q_\theta$ provides more refined phase beamforming that renders amplitude beamforming less effective.

\begin{figure}[!hhh] 
	\centering
	\scalebox{0.42}{\includegraphics{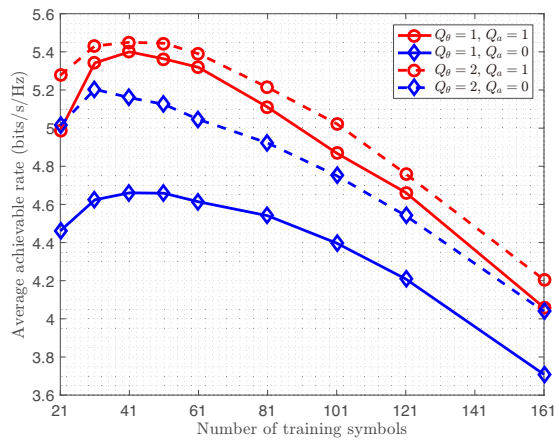}}
	\caption{Achievable rate versus number of uplink training symbols, $N_r$.}
	\label{figure_compare_training_time}
\end{figure}

\begin{figure}[!hhh] 
	\centering
	\scalebox{0.35}{\includegraphics{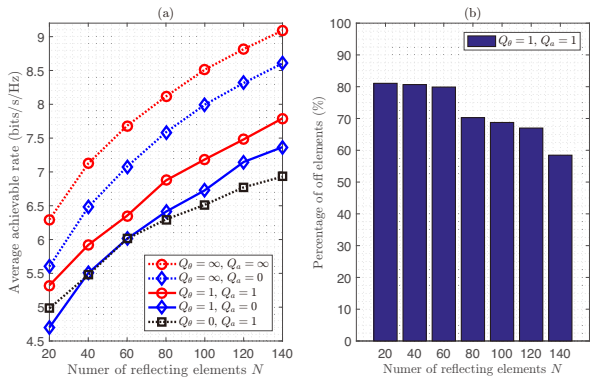}}
	\caption{Achievable rate and IRS EOP versus number of reflecting elements, $N$.}
	\label{Fig_SU_compare_N}
\end{figure}

\subsubsection{Impact of Number of Reflecting Elements, $N$} 
Next, in Fig. \ref{Fig_SU_compare_N}, we investigate the achievable rate  and IRS EOP versus the number of reflecting elements, $N$, with fixed $p_u = 10$ dBm. From Fig. \ref{Fig_SU_compare_N} (a), we observe that the performance gain of the proposed scheme with amplitude control is not sensitive to the value of $N$. This is mainly because in the interference power term $\mathbb{E} \{\tilde{\mathbf{v}}^H \Delta \tilde{\mathbf{H}}_k  \sum\nolimits_{j\in \mathcal{K}} \mathbf{w}_j \mathbf{w}_j^H \Delta \tilde{\mathbf{H}}_k^H \tilde{\mathbf{v}}\}$ (i.e., the term $(a)$ given in \eqref{effective_INR}), the channel MSE $\textrm{Tr}( \Delta \tilde {\mathbf{H}}_k \Delta \tilde {\mathbf{H}}_k^H)$ is inversely proportional to $N$, while with given $ \mathbf{A}_k=\Delta \tilde{\mathbf{H}}_k  \sum\nolimits_{j\in \mathcal{K}} \mathbf{w}_j \mathbf{w}_j^H \Delta \tilde{\mathbf{H}}_k^H$, $\tilde{\mathbf{v}}^H \mathbf{A}_k \tilde{\mathbf{v}}$ is proportional to $N$; as a result, the interference power does not scale with $N$. Besides, from Fig. \ref{Fig_SU_compare_N} (b), we can see that the IRS EOP decreases with the increasing of $N$. This is because when $N$ is small, the IRS can only perform very coarse phase beamforming, thus more elements should be turned off to control the interference due to CSI errors. For the same reason, the performance of amplitude beamforming is better than that of phase beamforming when $N$ is small.

\subsubsection{Impact of Maximum Downlink Transmit Power, $P$} 
Last, in Fig. \ref{Fig_SU_compare_P}, we plot the achievable rate and IRS EOP versus $P$ with $N=120$ and $p_u = 10$ dBm. We observe from Fig. \ref{Fig_SU_compare_P} (a) that the performance gain of the proposed scheme with amplitude control enlarges with the increasing of $P$. 
This is due to the fact that the  interference caused by imperfect CSI increases with  $P$ and thus the number of on elements of IRS needs to be reduced  in order to suppress the interference. This also explains why the IRS EOP increases with $P$, as shown in Fig. \ref{Fig_SU_compare_P} (b). In addition, we can observe that the performance achieved by amplitude beamforming improves faster than that of phase beamfoming as $P$ increases, since amplitude control is more useful in the large-$P$ regime.

\begin{figure}[!hhh] 
	\centering
	\scalebox{0.35}{\includegraphics{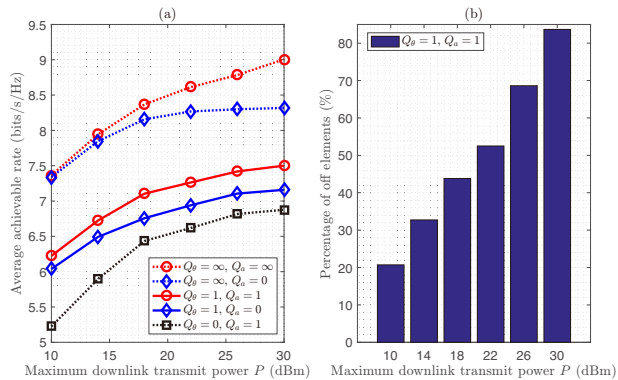}}
	\caption{Achievable rate and IRS EOP versus maximum downlink transmit power, $P$.}
	\label{Fig_SU_compare_P}
\end{figure}

\subsection{Multiuser System}

In this subsection, we consider a multiuser system with $K \geq 2$ users and the AP is equipped with $M=6$ antennas, with $p_u=18$ dBm and $N=60$. In Fig. \ref{Fig_MU}, we investigate the achievable sum-rate by the proposed PDD-based algorithm (i.e. Algorithm \ref{PDD_algorithm_multiuser}) versus the number of users, $K$. It is observed that the performance gain achieved by amplitude control is more pronounced when $K$ increases, since the multiuser interference due to  imperfect CSI becomes more severe as compared to the single-user case. Besides, the performance by amplitude beamforming is better than that by phase beamforming. Similar to the single-user case, this result indicates that in the case of imperfect CSI, it may be more beneficial to employ amplitude beamforming of lower cost as compared to phase beamforming, although this is not true for the case with more accurate CSI or perfect CSI.

\begin{figure}[!hhh] 
	\centering
	\scalebox{0.42}{\includegraphics{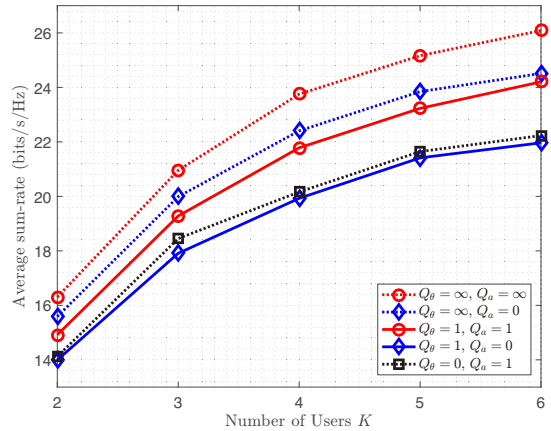}}
	\caption{Achievable sum-rate versus number of users, $K$.} 
	\label{Fig_MU}
\end{figure}

\section{Conclusions} \label{Section_conclusion}
In this paper, we studied the beamforming optimization in an IRS-aided multiuser system under imperfect CSI. Specifically, we first derived the distribution of CSI errors based on a practical time-varying reflection pattern based channel estimation method and then formulated the achievable rate maximization problem by jointly optimizing the reflection coefficients at the IRS as well as the transmit precoders at the AP. A penalized Dinkelbach-BSUM algorithm and a PDD-based algorithm were proposed for the single-user and multiuser cases, respectively. Simulation results showed that additional performance gains can be achieved by controlling the reflection amplitudes as compared to the case with full reflection/phase-shift control only, especially when the channel training resources are limited and/or the AP transmit power is high. It was also unveiled that IRS amplitude beamforming can be a low-cost alternative of the conventional phase beamforming in practice, yet offering comparable and even more favorable performance when CSI errors become a severe issue.

\begin{appendix}	
\subsection{Proof of Proposition \ref{achievable_rate_prop}} \label{appendix_achievable_rate}
According to \cite{Medard2000}, we can expand the mutual information based on differential entropy as follows:
\begin{equation} \label{deferential_entropy_interference}
I(s_k; y_k|\hat{\mathbf{H}}, \hat{\mathbf{h}}_d ) = h(s_k|\hat{\mathbf{H}}, \hat{\mathbf{h}}_d ) - h(s_k| y_k,\hat{\mathbf{H}}, \hat{\mathbf{h}}_d ).
\end{equation}
First, since $s_k$ is CSCG distributed with zero-mean and unit-variance, the first term on the right-hand-side (RHS) of \eqref{deferential_entropy_interference} becomes
$ \log (2\pi e)$. Then, using the fact that the entropy of a random variable with given variance is upper-bounded by the entropy of a Gaussian random variable with the same variance \cite{Medard2000}, the second term can be upper-bounded by the entropy of a Gaussian random variable with variance equal to\footnote{Here, we employ the following fact: $h(s_k| y_k,\hat{\mathbf{H}}, \hat{\mathbf{h}}_d ) = h(s_k-a_ky_k| y_k,\hat{\mathbf{H}}, \hat{\mathbf{h}}_d ) \leq h(s_k-a_ky_k|\hat{\mathbf{H}}, \hat{\mathbf{h}}_d ) $ holds for any $a_k$.}
\begin{equation}
\mathbb{E}_{\{s_k,\Delta \mathbf{H}_{k},\Delta \mathbf{h}_{d,k}, n_k\}}\left((s_k-g_k^H y_k) (s_k-g_k^H y_k)^H|\hat{\mathbf{H}}, \hat{\mathbf{h}}_d\right),
\end{equation}
where $g_k^Hy_k$ is the (linear) MMSE estimation of $s_k$. In the following, we equivalently rewrite $h(s_k| y_k,\hat{\mathbf{H}}, \hat{\mathbf{h}}_d)$ as $h(s_k| y_k,\bar{\mathbf{H}})$, with $\bar{\mathbf{H}}= \{\bar{\mathbf{H}}_{k}\}_{k\in\mathcal{K}}$. As a result, we can upper-bound $h(s_k| y_k,\bar{\mathbf{H}})$ as follows:
\begin{equation}
h(s_k| y_k,\bar{\mathbf{H}}) \leq  \log \left(2\pi e \left(1 - \frac{|\tilde{\mathbf{v}}^H \bar{\mathbf{H}}_k \mathbf{w}_k|^2}{\Psi_{k}^d + |\tilde{\mathbf{v}}^H \bar{\mathbf{H}}_{k}\mathbf{w}_k|^2}\right)\right),
\end{equation}
where $\tilde{\mathbf{v}} = [1, \mathbf{v}^T]^T$, and
\begin{equation} \label{effective_INR}
	\begin{aligned}
\Psi_{k}^d = & \underbrace{\mathbb{E} \{\tilde{\mathbf{v}}^H \Delta \tilde{\mathbf{H}}_k  \sum\limits_{j\in \mathcal{K}} \mathbf{w}_j \mathbf{w}_j^H \Delta \tilde{\mathbf{H}}_k^H \tilde{\mathbf{v}}\} }_{(a)} \\
&+ \sum\limits_{j \in \mathcal{K}\backslash k}  |\tilde{\mathbf{v}}^H \bar{\mathbf{H}}_{k} \mathbf{w}_j|^2 +  \sigma_k^2.
\end{aligned}
\end{equation}
To derive a tractable form of $(a)$ in \eqref{effective_INR}, we resort to the eigen-decomposition  $\sum_{j \in \mathcal{K}} \mathbf{w}_j \mathbf{w}_j^H  = \mathbf{Q} \mathbf{\Sigma}\mathbf{Q}^H$ and then  $(a)$ can be equivalently expressed as
\begin{equation}
(a) = \mathbb{E} \{\tilde{\mathbf{v}}^H \Delta \tilde{\mathbf{H}}_k  \mathbf{\Sigma} \Delta \tilde{\mathbf{H}}_k^H \tilde{\mathbf{v}}\} \overset{}{=} \mathbb{E} \left\{ \sum\limits_{m=1}^M  \lambda_m \tilde{\mathbf{v}}^H \bar{\mathbf{V}}^m_{k}\tilde{\mathbf{v}} \right\},
\end{equation}
where $\bar{\mathbf{V}}_{ij,k}^m = \dot{v}_{im,k} \dot{v}_{jm,k}^*$, 
$\dot{v}_{ij,k}$ represents the element on the $i$-th row and $j$-th column of the matrix $ \frac{1}{\sqrt{p_{u,k}}}  (\mathbf{V}^{\dagger})^H \mathbf{N}_{u,k}^H$ and $\lambda_i$ denotes the $i$-th diagonal element of $\mathbf{\Sigma}$. Furthermore, we observe that $\mathbb{E}\{\bar{\mathbf{V}}^m_{k} \} = \bar{\mathbf{V}}_k, \forall m$ and $\bar{\mathbf{V}}_{ij,k}=\frac{\varepsilon_{k}^2}{p_{u,k}} (\check{\mathbf{v}}_{:,i}^* \cdot \check{\mathbf{v}}_{:,j})$. 
Therefore, $\Psi_{k}^d $ can be equivalently expressed in the following deterministic form:
\begin{equation} 
\begin{aligned}
\Psi_{k}^d = & \sum\limits_{j \in \mathcal{K}\backslash k}  |(\mathbf{v}^H \hat{\mathbf{H}}_{k} + \hat{\mathbf{h}}_{d,k}^H) \mathbf{w}_j|^2\\
& +  \tilde{\mathbf{v}}^H \bar{\mathbf{V}}_k\tilde{\mathbf{v}} \sum\limits_{j \in \mathcal{K}}\|\mathbf{w}_j\|^2+  \sigma_k^2,
\end{aligned}
\end{equation}
which further leads to \eqref{inr_robust}.
Consequently, we can lower-bound $I(s_k; y_k|\hat{\mathbf{H}}, \hat{\mathbf{h}}_d )$ by \eqref{lower_bound}, which thus completes the proof.

\subsection{Proof of Proposition \ref{prop_convergent}} \label{appendix_C}
Let $f_n(\mathbf{v},\mathbf{u})$ and $f_d(\mathbf{v},\mathbf{u})$ denote the numerator and denominator of the objective function of problem \eqref{single_user_problem_d_eq3}, respectively. Also let $P(\mathbf{v},\mathbf{u},y) \triangleq f_n(\mathbf{v},\mathbf{u}  )-y\left(f_d(\mathbf{v},\mathbf{u})+\frac{1}{\beta}\| \mathbf{v} - \mathbf{u}\|^2\right)$ denote the objective value of problem \eqref{Dinkelbach_problem}. Then, we have
\begin{equation} \label{prop_convergent_1}
\begin{aligned}
 P(\mathbf{v}[i_d+1],& \mathbf{u}[i_d+1],y[i_d]) \\
 & \overset{(\textrm{i})}{\geq} P(\mathbf{v}[i_d+1],\mathbf{u}[i_d],y[i_d]) \\
& \overset{(\textrm{ii})}{\geq} P(\mathbf{v}[i_d],\mathbf{u}[i_d],y[i_d]) \\
& = P(\mathbf{v}[i_d+1],\mathbf{u}[i_d+1],y[i_d+1]) = 0,
\end{aligned}
\end{equation}
where $(\textrm{i})$ is due to the optimality of the $\mathbf{u}$-update step \eqref{u_update} and $(\textrm{ii})$ is due to the facts that the $\mathbf{v}$-update step \eqref{variable_update} is optimal for problem \eqref{subproblem1} and the objective function of \eqref{subproblem1} satisfies Assumption A in \cite{Hong2016}, i.e., the approximation function is a global upper bound of the objective function of \eqref{Dinkelbach_problem} (before ignoring the irrelevant terms) and their first-order derivatives are the same at the point of approximation; and the last two equalities are due to the Dinkelbach variable update step \eqref{dinkelbach_variable_update}. From \eqref{prop_convergent_1}, we have
\begin{equation} \label{prop_convergent_2}
\begin{aligned}
&f_n(\mathbf{v}[i_d+1],\mathbf{u}[i_d+1]  ) - y[i_d]\Big(f_d(\mathbf{v}[i_d+1],\mathbf{u}[i_d+1])\\
& +\frac{1}{\beta}\| \mathbf{v}[i_d+1] - \mathbf{u}[i_d+1]\|^2\Big)  \geq f_n(\mathbf{v}[i_d+1],\mathbf{u}[i_d+1]  )\\
& - y[i_d+1]\Big(f_d(\mathbf{v}[i_d+1],\mathbf{u}[i_d+1])+\frac{1}{\beta}\| \mathbf{v}[i_d+1] \\
& - \mathbf{u}[i_d+1]\|^2 \Big).
\end{aligned}
\end{equation}
Rearranging the terms in \eqref{prop_convergent_2} yields $y[i_d+1] \geq y[i_d]$. Therefore, $y$ is non-decreasing after each iteration, and since $y$ is also upper-bounded due to the constraints in problem \eqref{single_user_problem_d_relax}, the proposed penalized Dinkelbach-BSUM algorithm is guaranteed to converge with any given $\beta$. This thus completes the proof.

\subsection{BCD Algorithm  for Solving Problem \eqref{continuous_problem}} \label{appendix_BCD_continuous}
By focusing on the $n$-th reflecting element ($n \neq 1$), problem \eqref{continuous_problem} can be simplified as
\begin{equation} \label{n_subproblem}
	\begin{aligned}
\min \limits_{{v}_n}\; & \frac{-Pf_1(v_n)}{ Pf_2(v_n)+ \sigma^2} \\
\textrm{s.t.}\; & |v_n| \leq 1,
\end{aligned}
\end{equation}
where
\begin{equation} \label{no}
	\begin{aligned}
f_1(v_n)\triangleq & v_n^* \bm{\Phi}_{nn}  v_n +   \sum\limits_{j\in\tilde{\mathcal{N}} \backslash n} v_n^* \bm{\Phi}_{nj}  v_j \\
&+ \sum\limits_{i \in \tilde{\mathcal{N}} \backslash n} v_i^* \bm{\Phi}_{in}  v_n  + \sum\limits_{i\in\tilde{\mathcal{N}}\backslash n} \sum\limits_{j\in\tilde{\mathcal{N}} \backslash n} v_i^* \bm{\Phi}_{ij}  v_j,
\end{aligned}
\end{equation}
and $f_2(v_n)$ is similarly defined by replacing $\bm{\Phi}$ with $\bar{\mathbf{V}}$.
To derive the optimal solution of problem \eqref{n_subproblem}, we resort to the following Lagrangian function with $\lambda$ being the dual variable:
\begin{equation} \label{Lag}
\mathcal{L}(v_n,\lambda) = \frac{-Pf_1(v_n)}{ Pf_2(v_n)+ \sigma^2}  + \lambda(v_n^* v_n-1).
\end{equation}
Then, we consider the following three cases based on the value of $|v_n|$:
\begin{itemize}
	\item $|v_n|<1$: in this case, $\lambda=0$ must be satisfied due to the complementary slackness condition \cite{ConvexOptimization}. Let $x_n = \sum_{j\in\tilde{\mathcal{N}} \backslash n} \bm{\Phi}_{nj}  v_j$, $\tilde{x}_n = \sum_{j\in\tilde{\mathcal{N}} \backslash n} \bar{\mathbf{V}}_{nj}  v_j $, $z_n = \sum_{i\in\tilde{\mathcal{N}}\backslash n} \sum_{j\in\tilde{\mathcal{N}} \backslash n} v_i^* \bm{\Phi}_{ij}  v_j$ and $\tilde{z}_n = \sum_{i\in\tilde{\mathcal{N}}\backslash n} \sum_{j\in\tilde{\mathcal{N}} \backslash n} v_i^* \bar{\mathbf{V}}_{ij}  v_j$, and by resorting to the first-order optimality condition, we have
		\begin{equation} \label{v_n}
  \bar{a}_n v_n^2 + \bar{b}_n v_n + \bar{c}_n = 0,
	\end{equation}
	where $\bar{a}_n = P^2 \bm{\Phi}_{nn} \tilde{x}_n^* - P^2 \bar{\mathbf{V}}_{nn} x_n^*$, $\bar{b}_n = P^2 \bm{\Phi}_{nn} \tilde{z}_n+ P \bm{\Phi}_{nn} \sigma^2 + P^2 x_n \tilde{x}_n^* - P^2 \bar{\mathbf{V}}_{nn} z_n - P^2 \tilde{x}_n x_n^* $ and $\bar{c}_n = P^2 x_n \tilde{z}_n + P x_n \sigma^2 - P^2 \tilde{x}_n z_n$.
	Since \eqref{v_n} is a quadratic equation with respect to $v_n$, we can easily find its roots.
	\item $|v_n| = 1$: in this case, \eqref{Lag} reduces to 
	\begin{equation} \label{theta_n}
	\mathcal{L}(v_n,\lambda)  = -\frac{Pe^{\jmath \theta_n} x_n + Pe^{-\jmath \theta_n} x_n^* + \tilde{a}_n}{Pe^{\jmath \theta_n} \tilde{x}_n + Pe^{-\jmath \theta_n} \tilde{x}_n^* + \tilde{b}_n},
	\end{equation}
	where $\tilde{a}_n = P \bm{\Phi}_{nn} + P z_n$ and $\tilde{b}_n = P \bar{\mathbf{V}}_{nn} + P \tilde{z}_n + \sigma^2$. Let $\jmath P x_n \tilde{b}_n - \jmath P \tilde{x}_n \tilde{a}_n = x e^{\jmath \theta}$, $\tilde{c}_n = 2\Re\{ \jmath P^2 x_n \tilde{x}_n^*\} -  2\Re\{ \jmath P^2 \tilde{x}_n x_n^* \}$ and $v_n =e^{-\jmath \theta_n}$. By taking derivative of \eqref{theta_n} with respect to $\theta_n$, we obtain $2 x \cos(\theta + \theta_n) + \tilde{c}_n = 0 $, which further leads to $\theta_n = \arccos (-\frac{\tilde{c}_n}{2x}) -\theta$.
	\item $|v_n| = 0$: we have $v_n = 0$.
\end{itemize}
By comparing the objective values achieved by the above three cases, we can obtain the optimal solution of problem \eqref{continuous_problem}. Then, the proposed BCD algorithm for solving problem \eqref{continuous_problem} can be conducted by successively optimizing each reflection coefficient with the others being fixed.
\end{appendix}

\bibliographystyle{IEEETran}
\bibliography{references}

\begin{thebibliography}{10}
\providecommand{\url}[1]{#1}
\csname url@samestyle\endcsname
\providecommand{\newblock}{\relax}
\providecommand{\bibinfo}[2]{#2}
\providecommand{\BIBentrySTDinterwordspacing}{\spaceskip=0pt\relax}
\providecommand{\BIBentryALTinterwordstretchfactor}{4}
\providecommand{\BIBentryALTinterwordspacing}{\spaceskip=\fontdimen2\font plus
\BIBentryALTinterwordstretchfactor\fontdimen3\font minus
  \fontdimen4\font\relax}
\providecommand{\BIBforeignlanguage}[2]{{%
\expandafter\ifx\csname l@#1\endcsname\relax
\typeout{** WARNING: IEEEtran.bst: No hyphenation pattern has been}%
\typeout{** loaded for the language `#1'. Using the pattern for}%
\typeout{** the default language instead.}%
\else
\language=\csname l@#1\endcsname
\fi
#2}}
\providecommand{\BIBdecl}{\relax}
\BIBdecl

\bibitem{Zhaoglobecom2020}
M.~M. {Zhao}, Q.~{Wu}, M.~J. {Zhao}, and R.~{Zhang}, ``{IRS}-aided wireless
  communication with imperfect {CSI}: {Is} amplitude control helpful or not?''
  in \emph{Proc. IEEE Global Communications Conference (GLOBECOM)}, Dec. 2020,
  pp. 1--6.

\bibitem{Boccardi2014}
F.~{Boccardi}, R.~W. {Heath}, A.~{Lozano}, T.~L. {Marzetta}, and P.~{Popovski},
  ``Five disruptive technology directions for {5G},'' \emph{IEEE Commun. Mag.},
  vol.~52, no.~2, pp. 74--80, Feb. 2014.

\bibitem{Wu2019Magazine}
Q.~{Wu} and R.~{Zhang}, ``Towards smart and reconfigurable environment:
  {I}ntelligent reflecting surface aided wireless network,'' \emph{IEEE Commun.
  Mag.}, vol.~58, no.~1, pp. 106--112, Jan. 2020.

\bibitem{Wu2018_journal}
{Q. {Wu} and R. {Zhang}}, ``Intelligent reflecting surface enhanced wireless
  network via joint active and passive beamforming,'' \emph{IEEE Trans.
  Wireless Commun.}, vol.~18, no.~11, pp. 5394--5409, Nov. 2019.

\bibitem{Basar2019}
E.~{Basar}, M.~{Di Renzo}, J.~{De Rosny}, M.~{Debbah}, M.~{Alouini}, and
  R.~{Zhang}, ``Wireless communications through reconfigurable intelligent
  surfaces,'' \emph{IEEE Access}, vol.~7, pp. 116\,753--116\,773, Aug. 2019.

\bibitem{Huang2019}
C.~{Huang}, A.~{Zappone}, G.~C. {Alexandropoulos}, M.~{Debbah}, and C.~{Yuen},
  ``Reconfigurable intelligent surfaces for energy efficiency in wireless
  communication,'' \emph{IEEE Trans. Wireless Commun.}, vol.~18, no.~8, pp.
  4157--4170, Aug. 2019.

\bibitem{cui2014coding}
T.~J. Cui, M.~Q. Qi, X.~Wan, J.~Zhao, and Q.~Cheng, ``Coding metamaterials,
  digital metamaterials and programmable metamaterials,'' \emph{L. Sci. \&
  Appl.}, vol.~3, no.~10, pp. e218--e218, Oct. 2014.

\bibitem{Yang2019}
Y.~{Yang}, B.~{Zheng}, S.~{Zhang}, and R.~{Zhang}, ``Intelligent reflecting
  surface meets {OFDM:} {P}rotocol design and rate maximization,'' \emph{IEEE
  Trans. Commun.}, vol.~68, no.~7, pp. 4522--4535, Jul. 2020.

\bibitem{Cui2019}
M.~{Cui}, G.~{Zhang}, and R.~{Zhang}, ``Secure wireless communication via
  intelligent reflecting surface,'' \emph{IEEE Wireless Commun. Lett.}, vol.~8,
  no.~5, pp. 1410--1414, Oct 2019.

\bibitem{Jiang2019}
T.~{Jiang} and Y.~{Shi}, ``Over-the-air computation via intelligent reflecting
  surfaces,'' in \emph{Proc. IEEE GLOBECOM}, Dec. 2019, pp. 1--6.

\bibitem{zuo2020resource}
J.~{Zuo}, Y.~{Liu}, Z.~{Qin}, and N.~{Al-Dhahir}, ``Resource allocation in
  intelligent reflecting surface assisted {NOMA} systems,'' \emph{IEEE Trans.
  Commun.}, vol.~68, no.~11, pp. 7170--7183, Nov. 2020.

\bibitem{ZhangMIMO}
S.~Zhang and R.~Zhang, ``Capacity characterization for intelligent reflecting
  surface aided {MIMO} communication,'' \emph{IEEE J. Sel. Areas Commun.},
  vol.~38, no.~8, pp. 1823--1838, Aug. 2020.

\bibitem{Mishra2019ICASSP}
D.~{Mishra} and H.~{Johansson}, ``Channel estimation and low-complexity
  beamforming design for passive intelligent surface assisted {MISO} wireless
  energy transfer,'' in \emph{Proc. IEEE International Conference on Acoustics,
  Speech and Signal Processing (ICASSP)}, May 2019, pp. 4659--4663.

\bibitem{zheng2019intelligent}
B.~{Zheng} and R.~{Zhang}, ``Intelligent reflecting surface-enhanced {OFDM}:
  {C}hannel estimation and reflection optimization,'' \emph{IEEE Wireless
  Commun. Lett.}, vol.~9, no.~4, pp. 518--522, Apr. 2020.

\bibitem{you2019progressive}
C.~You, B.~Zheng, and R.~Zhang, ``Channel estimation and passive beamforming
  for intelligent reflecting surface: {D}iscrete phase shift and progressive
  refinement,'' \emph{IEEE J. Sel. Areas Commun.}, vol.~38, no.~11, pp.
  2604--2620, Nov. 2020.

\bibitem{wang2019channel}
Z.~{Wang}, L.~{Liu}, and S.~{Cui}, ``Channel estimation for intelligent
  reflecting surface assisted multiuser communications: Framework, algorithms,
  and analysis,'' \emph{IEEE Trans. Wireless Commun.}, vol.~19, no.~10, pp.
  6607--6620, Oct. 2020.

\bibitem{zheng2020intelligent}
B.~{Zheng}, C.~{You}, and R.~{Zhang}, ``Intelligent reflecting surface assisted
  multi-user {OFDMA}: {C}hannel estimation and training design,'' \emph{IEEE
  Trans. Wireless Commun.}, vol.~19, no.~12, pp. 8315--8329, Dec. 2020.

\bibitem{He2019_CE}
Z.~{He} and X.~{Yuan}, ``Cascaded channel estimation for large intelligent
  metasurface assisted massive {MIMO},'' \emph{IEEE Wireless Commun. Lett.},
  vol.~9, no.~2, pp. 210--214, Feb. 2020.

\bibitem{chen2019channel}
J.~Chen, Y.-C. Liang, H.~V. Cheng, and W.~Yu, ``Channel estimation for
  reconfigurable intelligent surface aided multi-user {MIMO} systems,''
  \emph{arXiv preprint arXiv:1912.03619}, 2019.

\bibitem{yu2019robust}
X.~{Yu}, D.~{Xu}, Y.~{Sun}, D.~W.~K. {Ng}, and R.~{Schober}, ``Robust and
  secure wireless communications via intelligent reflecting surfaces,''
  \emph{IEEE J. Sel. Areas Commun.}, vol.~38, no.~11, pp. 2637--2652, Nov.
  2020.

\bibitem{zhou2019robust}
G.~{Zhou}, C.~{Pan}, H.~{Ren}, K.~{Wang}, M.~{Di Renzo}, and A.~{Nallanathan},
  ``Robust beamforming design for intelligent reflecting surface aided {MISO}
  communication systems,'' \emph{IEEE Wireless Commun. Lett.}, vol.~9, no.~10,
  pp. 1658--1662, Oct. 2020.

\bibitem{zhou2020framework}
G.~{Zhou}, C.~{Pan}, H.~{Ren}, K.~{Wang}, and A.~{Nallanathan}, ``A framework
  of robust transmission design for {IRS}-aided {MISO} communications with
  imperfect cascaded channels,'' \emph{IEEE Trans. Signal Process.}, vol.~68,
  pp. 5092--5106, 2020.

\bibitem{Wu2019Discrete}
Q.~{Wu} and R.~{Zhang}, ``Beamforming optimization for wireless network aided
  by intelligent reflecting surface with discrete phase shifts,'' \emph{IEEE
  Trans. Commun.}, vol.~68, no.~3, pp. 1838--1851, Mar. 2020.

\bibitem{guo2019weighted}
H.~{Guo}, Y.~{Liang}, J.~{Chen}, and E.~G. {Larsson}, ``Weighted sum-rate
  maximization for reconfigurable intelligent surface aided wireless
  networks,'' \emph{IEEE Trans. Wireless Commun.}, vol.~19, no.~5, pp.
  3064--3076, May 2020.

\bibitem{Bertsekas1999}
D.~Bertsekas, \emph{Nonlinear Programming}.\hskip 1em plus 0.5em minus
  0.4em\relax 2nd ed. Belmont, MA: Athena Scientific, 1999.

\bibitem{dinkelbach1967nonlinear}
W.~Dinkelbach, ``On nonlinear fractional programming,'' \emph{Manage. Sci.},
  vol.~13, no.~7, pp. 492--498, Mar. 1967.

\bibitem{Hong2016}
M.~{Hong}, M.~{Razaviyayn}, Z.~{Luo}, and J.~{Pang}, ``A unified algorithmic
  framework for block-structured optimization involving big data: {W}ith
  applications in machine learning and signal processing,'' \emph{IEEE Signal
  Process. Mag.}, vol.~33, no.~1, pp. 57--77, Jan. 2016.

\bibitem{shi2017penalty}
Q.~{Shi} and M.~{Hong}, ``Penalty dual decomposition method for nonsmooth
  nonconvex optimization—{Part I}: {A}lgorithms and convergence analysis,''
  \emph{IEEE Trans. Signal Process.}, vol.~68, pp. 4108--4122, 2020.

\bibitem{zhao2019intelligent}
M.~M. {Zhao}, Q.~{Wu}, M.~J. {Zhao}, and R.~{Zhang}, ``Intelligent reflecting
  surface enhanced wireless network: {Two}-timescale beamforming
  optimization,'' \emph{IEEE Trans. Wireless Commun.}, vol.~20, no.~1, pp.
  2--17, Jan. 2021.

\bibitem{PhysRevApplied}
{F. Liu \emph{et al.}}, ``Intelligent metasurfaces with continuously tunable
  local surface impedance for multiple reconfigurable functions,'' \emph{Phys.
  Rev. Applied}, vol.~11, p. 044024, Apr. 2019.

\bibitem{Nayeri2018book}
P.~Nayeri, F.~Yang, and A.~Z. Elsherbeni, \emph{Reflectarray antennas:
  {T}heory, designs, and applications}.\hskip 1em plus 0.5em minus 0.4em\relax
  John Wiley \& Sons, 2018.

\bibitem{Abeywickrama2020}
S.~{Abeywickrama}, R.~{Zhang}, Q.~{Wu}, and C.~{Yuen}, ``Intelligent reflecting
  surface: {P}ractical phase shift model and beamforming optimization,''
  \emph{IEEE Trans. Commun.}, vol.~68, no.~9, pp. 5849--5863, Sep. 2020.

\bibitem{jung2019optimality}
M.~Jung, W.~Saad, M.~Debbah, and C.~S. Hong, ``On the optimality of
  reconfigurable intelligent surfaces {(RISs)}: {P}assive beamforming,
  modulation, and resource allocation,'' \emph{arXiv preprint
  arXiv:1910.00968}, 2019.

\bibitem{Wu2020tutorial}
Q.~{Wu}, S.~{Zhang}, B.~{Zheng}, C.~{You}, and R.~{Zhang}, ``Intelligent
  reflecting surface aided wireless communications: {A} tutorial,'' \emph{IEEE
  Trans. Commun.}, DOI: 10.1109/TCOMM.2021.3051897, 2021.

\bibitem{LuoSDR2010}
Z.~{Luo}, W.~{Ma}, A.~M. {So}, Y.~{Ye}, and S.~{Zhang}, ``Semidefinite
  relaxation of quadratic optimization problems,'' \emph{IEEE Signal Process.
  Mag.}, vol.~27, no.~3, pp. 20--34, May 2010.

\bibitem{Wang2014}
K.~{Wang}, A.~M. {So}, T.~{Chang}, W.~{Ma}, and C.~{Chi}, ``Outage constrained
  robust transmit optimization for multiuser {MISO} downlinks: {T}ractable
  approximations by conic optimization,'' \emph{IEEE Trans. Signal Process.},
  vol.~62, no.~21, pp. 5690--5705, Nov. 2014.

\bibitem{ConvexOptimization}
S.~Boyd and L.~Vandenberghe, \emph{Convex Optimization}.\hskip 1em plus 0.5em
  minus 0.4em\relax Cambridge, U.K.: Cambridge Univ. Press, 2004.

\bibitem{Shi2011WMMSE}
Q.~{Shi}, M.~{Razaviyayn}, Z.~{Luo}, and C.~{He}, ``An iteratively weighted
  {MMSE} approach to distributed sum-utility maximization for a {MIMO}
  interfering broadcast channel,'' \emph{IEEE Trans. Signal Process.}, vol.~59,
  no.~9, pp. 4331--4340, Sep. 2011.

\bibitem{Zhao2019CL}
M.~{M. Zhao}, Q.~{Shi}, Y.~{Cai}, M.~{J. Zhao}, and Q.~{Yu}, ``Decoding binary
  linear codes using penalty dual decomposition method,'' \emph{IEEE Commun.
  Lett.}, vol.~23, no.~6, pp. 958--962, Jun. 2019.

\bibitem{CVX}
M.~Grant and S.~Boyd, ``{CVX}: Matlab software for disciplined convex
  programming, version 2.1,'' \url{http://cvxr.com/cvx}, Mar. 2014.

\bibitem{Medard2000}
M.~{Medard}, ``The effect upon channel capacity in wireless communications of
  perfect and imperfect knowledge of the channel,'' \emph{IEEE Trans. Inf.
  Theory}, vol.~46, no.~3, pp. 933--946, May 2000.

\end{thebibliography}

\end{document}